\documentclass[journal,12pt,draftclsnofoot,onecolumn]{IEEEtran}
\usepackage{amssymb}
\usepackage{amsmath}
\usepackage{cite}
\usepackage{url}
\usepackage{xcolor}
\usepackage{citesort}
\usepackage{cite,graphicx,amsmath,amssymb}
\usepackage{fancyhdr}
\usepackage{mdwmath}
\usepackage{mdwtab}
\usepackage{caption}
\usepackage{amsthm}
\usepackage{setspace}
\usepackage{algorithm}
\usepackage{algorithmic}
\usepackage{makecell}
\usepackage{diagbox}

\usepackage{mathtools}
\usepackage{amsfonts}
\usepackage{algorithm}
\usepackage{caption}
\usepackage{algorithmic}
\usepackage{subfig}

\newtheorem{remark}{Remark}
\newtheorem{theorem}{Theorem}

\newtheorem{lemma}{Lemma}

\newtheorem{corollary}{Corollary}

\newtheorem{proposition}{Proposition}
\allowdisplaybreaks
\makeatletter
\def\ScaleIfNeeded{%
\ifdim\Gin@nat@width>\linewidth \linewidth \else \Gin@nat@width
\fi } \makeatother
\begin{document}
\title{Resource Allocation in STAR-RIS-Aided Networks: OMA and NOMA}
% \markboth{\textit{A Manuscript Submitted to The IEEE Communications Magazine} }

\author{
Chenyu Wu, Xidong Mu, \IEEEmembership{Graduate Student Member, IEEE}, Yuanwei~Liu, \IEEEmembership{Senior Member,~IEEE},  Xuemai Gu, \IEEEmembership{Member,~IEEE}, and Xianbin Wang, \IEEEmembership{Fellow,~IEEE}
\thanks{C. Wu and X. Gu are with the School of Electronic and Information Engineering, Harbin Institute of Technology (HIT), Harbin, 150001, China. (e-mail: \{wuchenyu, guxuemai\}@hit.edu.cn).}
\thanks{X. Mu is with the School of Artificial Intelligence, Beijing University of Posts and Telecommunications, Beijing, 100876, China (email: muxidong@bupt.edu.cn).}
\thanks{Y. Liu is with the School of Electronic Engineering and Computer Science, Queen Mary University of London, London E1 4NS, U.K. (email: yuanwei.liu@qmul.ac.uk).}
\thanks{X. Wang is with the Department of Electrical and Computer Engineering, Western University, London, ON N6A 5B9, Canada (e-mail: xianbin.wang@uwo.ca).}
}

\maketitle

%\begin{abstract}
%This article focuses on the exploitation of reconfigurable intelligent surfaces (RISs) in multi-user networks employing orthogonal multiple access (OMA) or non-orthogonal multiple access (NOMA), with an emphasis on investigating the interplay between NOMA and RIS. Depending on whether the RIS reflection coefficients can be adjusted only once or multiple times during one transmission, we distinguish between \emph{static} and \emph{dynamic} RIS configurations. In particular, the capacity region of RIS aided single-antenna NOMA networks is characterized and compared with the OMA rate region from an information-theoretic perspective, revealing that the dynamic RIS configuration is capacity-achieving. Then, the impact of the RIS deployment location on the performance of different multiple access schemes is investigated, which reveals that asymmetric and symmetric deployment strategies are preferable for NOMA and OMA, respectively. Furthermore, for RIS aided multiple-antenna NOMA networks, three novel joint active and passive beamformer designs are proposed based on both beamformer based and cluster based strategies. Finally, open research problems for RIS-NOMA networks are highlighted.
%\end{abstract}
%\begin{IEEEkeywords}
%Heterogeneous ultra dense networks, non-orthogonal multiple access, massive connectivity, user association, and resource allocation.
%\end{IEEEkeywords}

%
\vspace{-1.2cm}
\begin{abstract}	
Simultaneously transmitting and reflecting reconfigurable intelligent
surface (STAR-RIS) is a promising technology to achieve full-space coverage by splitting the incident signal into the transmitted and reflected signals towards both sides of the surface. This paper investigates the resource allocation problem in a STAR-RIS-assisted multi-carrier communication network. To maximize the system sum-rate, a joint optimization problem of channel assignment, power allocation, and transmission- and reflection-beamforming at the STAR-RIS for orthogonal multiple access (OMA) is first formulated, which is a mixed-integer non-linear programming problem. To solve this challenging problem, we first propose a channel assignment scheme utilizing matching theory and then invoke the alternating optimization-based method to optimize the resource allocation policy and beamforming vectors iteratively. Furthermore, the sum-rate maximization problem for non-orthogonal multiple access (NOMA) with flexible decoding orders is investigated. To efficiently solve it, we first propose a location-based matching algorithm to determine the sub-channel assignment, where a transmitted user and a reflected user are grouped on a sub-channel. Based on this \emph{transmission-and-reflection} sub-channel assignment strategy, a three-step approach is proposed, where the decoding orders, beamforming-coefficient vectors, and power allocation are optimized by employing semidefinite programming, convex upper bound approximation, and geometry programming, respectively. Numerical results unveil that: 1) For OMA, a general design that includes same-side user-pairing for channel assignment is preferable, while for NOMA, the proposed transmission-and-reflection scheme can achieve comparable performance as the exhaustive search-based algorithm. 2) The STAR-RIS-aided NOMA network significantly outperforms the networks employing conventional RISs and OMA. %3) The performance gain of NOMA over OMA is significantly enhanced by deploying the STAR-RIS.

\end{abstract}

\section{Introduction}
%The emergence of killer applications and the explosive growth of data traffic have promoted the revolution of the sixth-generation (6G) wireless systems, which necessitates novel communication paradigms to match the new technical requirements\cite{6g1,6g2}. Great research interest has been witnessed in new technologies, such as massive multiple-input multiple-output (MIMO), and TeraHertz (THz) communications\cite{6g3}, which have been proved to be effective to some extent. However, due to the randomness of radio environment, the propagation of wireless signals still suffers from fading caused by random scattering, reflection, and refraction, which degrades the quality of the received signals and limits the performance of wireless networks.

%it is still challenging to realize these technologies due to the severe signal deterioration, high operational costs, and inter/intra cell interference. 

The emergence of killer applications and the explosive growth of data traffic have promoted the revolution of the sixth-generation (6G) wireless systems, which necessitates novel communication paradigms to match the new technical requirements\cite{6g1,6g2}. With the development of metasurfaces, reconfigurable intelligent surfaces (RISs) have been proposed as a promising candidate to improve the performance of wireless communication networks and give rise to the concept of ``smart radio environments'' (SREs)\cite{3,4}. An RIS is a planar surface of massive elements made from electromagnetic material. %Different from other mentioned paradigms, which still leave the signal propagation uncontrollable, RISs can turn the wireless environment into a reconfigurable entity, which can be optimized according to desired objectives. 
By controlling the phase and amplitude responses of these elements, RISs can effectively adjust the propagation of the incident wireless signals to create desirable multi-path effects, such as enhancing the desired received signals or eliminating inter/intra-cell interference\cite{survey_ris}. Moreover, due to the nearly-passive nature, RISs can be deployed to improve the spectrum efficiency and the reliability of wireless systems using much lower energy consumption compared to active transceivers. Given the aforementioned appealing characteristics, RISs have aroused great interest in both academia and industry.

However, most existing works assume that RISs can only reflect the incident signal\cite{simulation,ofdm,uav,swipt,nomaris1,nomaris2,mu1,nomaris3,nomaris5,nomaris4,nomaris6}. As a result, the communication system can not exploit the benefits of the RIS when the transmitter and receiver are not located on the same side of the RIS. To overcome this drawback, a novel concept of simultaneously transmitting and reflecting RISs (STAR-RISs)\cite{STAR}, also called intelligent omni-surfaces (IOS)\cite{ios}, is proposed. Different from conventional reflecting-only RISs, each element of STAR-RISs can transmit and reflect the incident signal simultaneously, thus breaking the geographical limitations of deploying RISs and achieving \emph{full-space coverage}. Moreover, STAR-RISs offer new degrees-of-freedom (DoF) to SREs since both the transmission- and reflection-coefficients can be designed to further enhance the performance of wireless networks.

Inspired by the benefits of STAR-RISs, it is promising to employ STAR-RIS to improve the performance of communication systems. However, to guarantee the high performance of the next-generation communication system, it is still necessary to design a transmission scheme to meet the heterogeneous demands on high throughput and reliability. Existing multiple access
techniques mainly contain two categories: orthogonal multiple access (OMA) and non-orthogonal multiple access (NOMA). Unlike OMA, where orthogonal time/frequency resource blocks are utilized to serve each user, NOMA allows multiple users to share the same resource blocks by employing superposition coding at the transmitters and successive interference cancellation (SIC) at the receivers\cite{proceeding}. Owing to this unique feature, NOMA has been regarded as an enabling technology to improve spectral efficiency and support massive connectivity\cite{noma,proceeding} compared to OMA. 
Driven by the aforementioned issues, this paper aims to investigate the promising application of STAR-RISs in wireless networks considering different transmission schemes.

%It is worth mentioning that besides multiplexing
%gain brought by NOMA, multi-carrier systems can providing additional degrees of freedom offered by multiuser diversity, which motivates us to study multi-carrier NOMA systems. Inspired by the aforementioned benefits, in this paper, we aim to investigate the promising combination of STAR-RISs and multi-carrier NOMA techniques to improve the performance of communication systems.

\subsection{Prior Works}

\emph{1) Studies on RIS-Assisted Wireless Systems:}
There have been extensive works on various applications of RISs in wireless systems, such as the RIS-aided multiple-input-single-output (MISO)\cite{simulation}, RIS-aided orthogonal frequency division multiplexing (OFDM)\cite{ofdm}, RIS-aided unmanned aerial
vehicle (UAV) systems\cite{uav}, RIS-aided simultaneous wireless information and power transfer (SWIPT) systems\cite{swipt}. Nevertheless, these research contributions mainly considered single user scenario or the OMA scheme. Driven by the benefits of RIS and NOMA, their combination has been investigated in many prior works\cite{nomaris1,mu1,nomaris2,nomaris3,nomaris5,nomaris4,nomaris6}. Specifically, the authors of \cite{nomaris1} considered an RIS-aided NOMA system with single-input-single-output (SISO) and MISO setup, respectively, where the formulated minimum rate maximization problem was solved optimally using an iterative algorithm. The total transmit power was minimized in an RIS-empowered MISO-NOMA network by designing the joint beamforming using an efficient difference-of-convex method\cite{nomaris2}. In \cite{mu1}, an RIS-assisted MISO-NOMA system was studied, where the joint active and passive beamforming was designed by alternating optimization. %Extended works considered scenarios including multi-channel, multi-cell,  and multi-cluster modes. 
The authors of\cite{nomaris3} presented a multi-channel RIS-NOMA framework to study the throughput maximization problem, in which the channel assignment, decoding order, power allocation, and reflection coefficients were jointly optimized. The authors of \cite{nomaris5} further considered the multi-cell RIS-aided NOMA networks and proposed a resource allocation algorithm using convex optimization methods and matching theory. For the multi-cluster RIS-sided MISO-NOMA scenarios, to minimize the total transmit power, the beamforming vectors at the base station (BS) and the phase shift design at the RIS were jointly optimized in\cite{nomaris4,nomaris6}.

\emph{2) Studies on STAR-RIS-Aided Communication Systems:} There have been some initial works studying potential profits of deploying STAR-RIS in wireless communication systems. Specifically, in \cite{SRAR}, the hardware model of STAR-RISs was presented and the diversity gain over conventional RISs was analyzed based on the proposed channel models.  In \cite{star2}, the coverage characterization of STAR-RIS-aided communication networks was investigated, and the formulated total coverage range maximization problems for NOMA and OMA networks were solved.
A STAR-RIS-aided communication system for both unicast and multicast transmission was considered in\cite{star1}, where three practical operating protocols for STAR-RISs were proposed and the corresponding joint active and passive beamforming for minimizing transmit power was studied. An IOS-assisted communication system was considered in \cite{ios}, where the formulated sum-rate maximization problem was solved by jointly designing phase shifts at the IOS and beamforming at the BS. 
\subsection{Motivations and Contributions}
%The advantages of deploying STAR-RIS to assist wireless communication systems are twofold: 
%1) STAR-RISs have both transmission and reflection functions, which provides ubiquitous wireless coverage for users located at both sides. 2) STAR-RISs provide new DoFs of network resource allocation design, such as SIC decoding orders and passive beamforming matrix optimization. 
%
%It is also known that NOMA transmission benefits a lot from asymmetric channels. 

Compared with conventional reflecting-only RISs\cite{simulation,ofdm,uav,swipt,nomaris1,nomaris2,mu1,nomaris3,nomaris4,nomaris5,nomaris6}, a full-space SRE can be achieved by deploying STAR-RISs\cite{STAR}. To fully reap the potentials of STAR-RISs, it is fundamental to design the resource allocation policy based on different multiple access schemes. However, the formulated problem for OMA is non-trivial to solve since the variables like transmit power and transmission- and reflection-beamforming vectors are highly coupled. For NOMA, the problem is more challenging due to the flexible decoding orders and the existence of co-channel interference. Therefore, it is of vital importance to develop efficient algorithms for STAR-RIS-aided wireless communication systems. Sparked by the above backgrounds, in this paper, we investigate the efficient design of  STAR-RIS-aided OMA and NOMA networks, respectively. The channel assignment, resource allocation policy at the access point (AP), and transmission- and reflection-beamforming at the STAR-RIS are jointly optimized to enhance the system performance. To the best of our knowledge, the research on resource allocation strategy in STAR-RIS-aided communication networks is still in its infancy.

In light of the above motivations and challenges, the main contributions of this paper are summarized as follows:

\begin{itemize} %[]里面加入符号
	\item We consider a STAR-RIS-assisted multi-carrier downlink communication network, where an AP sends information to multiple users with the aid of a STAR-RIS. We formulate a sum-rate maximization problem by jointly optimizing the channel assignment, transmission- and reflection-coefficient vector, power allocation, and time allocation factor/decoding order for OMA and NOMA, respectively, which is mixed-integer non-linear programming (MINLP).
	
	\item For OMA, the channel assignment strategy is first designed using matching theory. Then, we propose an efficient algorithm by invoking alternating optimization to solve the resulting resource allocation problem, where the time allocation factor, power allocation, and transmission- and reflection-coefficient vector are optimized iteratively.	
	\item For NOMA, we first propose a low-complexity location-based matching algorithm (LMA) to deal with the channel assignment, where a transmission user and a reflection user are selected to be paired on a sub-channel. With the derived channel assignment, an efficient three-stage algorithm is developed. First, we obtain the NOMA decoding orders by finding the beamforming vectors that lead to the maximum combined channel gain. Then, the transmission- and reflection-coefficient vectors are optimized using convex upper bound (CUB) approximation. Finally, the problem of power allocation is reformulated as geometry programming (GP), which can be solved optimally.
	\item Numerical results demonstrate that, for OMA, the proposed design can achieve near-optimal performance when considering all possibilities of channel assignment. For NOMA, the proposed LMA yields a good complexity-optimality tradeoff thanks to the transmission-reflection user-pairing scheme. Moreover, the sum-rate achieved by the STAR-RIS-NOMA network is more significant than that achieved by STAR-OMA and conventional RIS-aided networks. 
\end{itemize}
 
\subsection{Organizations and Notations}
The remainder of this paper is organized as follows. In Section II, we present the system model and problem formulation. In Sections III and IV, efficient algorithms are proposed to solve the resource allocation problem for OMA and NOMA, respectively. Simulation results are presented in Section V, which is followed by conclusions in Section VI.

\emph{Notations:} In this paper, vectors and matrices are denoted as boldface lower-case and boldface
capital letters, respectively. $\mathbb{C}^{M\times1}$ denotes the space of $M\times1$ complex vectors. diag($\mathbf{x}$) denotes a diagonal matrix whose diagonal elements
are the corresponding elements in vector $\mathbf{x}$. Diag($\mathbf{A}$) denotes a vector whose elements are extracted from the main diagonal of matrix $\mathbf{A}$.
$\mathbf{x}^H$ denotes the conjugate transpose of vector $\mathbf{x}$. Tr($\mathbf{X}$) and rank($\mathbf{X}$) denote the trace and rank of matrix $\mathbf{X}$, respectively. real($x$) and imag($x$) denote the real and imaginary
part of a complex number $x$, respectively. $\triangleq$ stands for definition.
\section{System Model and Problem Formulation}

\begin{figure}[t!]
  \centering
  \includegraphics[width=0.6\textwidth]{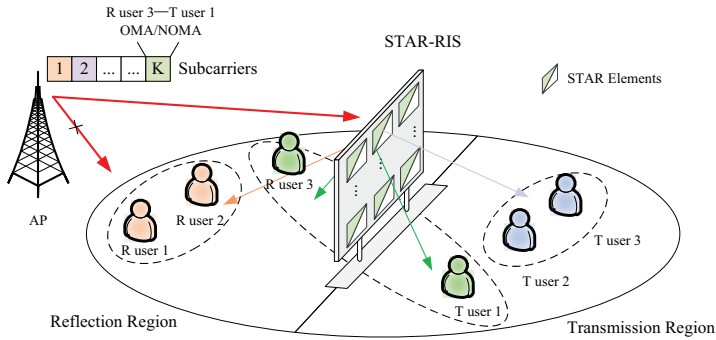}\\
  \caption{Illustration of a STAR-RIS aided downlink multi-carrier communication network.}
\end{figure}
\subsection{System Model}
As illustrated in Fig. 1, we consider a STAR-RIS-assisted downlink communication system, where a single-antenna AP communicates with $I$ single-antenna users with the aid of a STAR-RIS composed of $M$ elements. 
%Since the STAR-RIS can realize full-space coverage by simultaneously transmitting and reflecting the incident signal, we assume that a set $\mathcal{I}$ of $I$ users are evenly located at the two sides of the STAR-RIS. We refer to users that are located behind the STAR-RIS (i.e., transmission region) as T users, denoted as $\mathcal{I}_t=\{1,...,I_t\}$, while users that are located in front of the STAR-RIS (i.e., reflection region) are referred to as R users, denoted as $\mathcal{I}_r=\{1,...,I_r\}$. 
At the AP, the total bandwidth $W$ is divided into $K$ sub-channels, denoted by $\mathcal{K}=\{1,...,K\}$, and each sub-channel can be occupied by at most two users to improve the spectrum efficiency. 
%\footnote{In this paper, to get design insights of NOMA user-pairing, we assume $N_\text{max}=2$. The proposed algorithm can be easily extended to systems with $N_\text{max}>2$.} 
Moreover, we assume that each user $i\in\mathcal{I}=\{1,...,I\}$ can occupy at most one subcarrier. Hence, all $I$ users are divided into $K$
pairs, i.e., $I=2K$.  Let binary variable $\lambda_{k,i}$ denote whether the $k$-th subchannel is assigned to user $i$ and $\mathcal{I}_k$ denote the
collection of the user-pair on sub-channel $k$, i.e., $\mathcal{I}_k=\{i|\lambda_{k,i}=1,\forall{i}\in\mathcal{I}\}$,$\forall{k}\in\mathcal{K}$. In this paper, we assume that perfect channel state information (CSI) is available at the AP to study the maximum performance gain.

The incident signals impinged on each element of a STAR-RIS can be split into transmitted and reflected signals to reach users at both sides of the surface. As shown in Fig. 1, the left-half part is known as the reflection region, while the right-half part is referred to as the transmission region. In this paper, we assume that users are equally located in these two regions. Further, users that are located at the transmission region are referred to as T users, whose set is denoted by $\mathcal{T}=\{T_1,...,T_K\}$. Similarly, users that are located at the reflection region are called R users, whose set is $\mathcal{R}=\{R_1,...,R_K\}$. 

\subsection{STAR-RIS Model}

 %In this paper, we consider the energy splitting protocol proposed in , where all elements of the STAR-RIS can operate in general simultaneous transmission- and reflection-mode. 
To mathematically characterize the operation model of the STAR-RIS, we define transmission- and reflection-coefficient vectors (also called beamforming vectors for simplification) as\cite{star1}
%\begin{equation}
%	\begin{aligned}
%			\boldsymbol{v}_n&\triangleq	\left\{\begin{gathered}
%				\begin{aligned}
%						&\boldsymbol{v}_t,\ \ \ {\text{if user $i$ is in transmission region,}}\\
%				&{{\mathbf{v}}_r}, \ \ \ {\text{if user $i$ is in reflection region,}} \\
%				\end{aligned}
%			\end{gathered}\right.
%\\	&=\left\{\begin{gathered}
%	 [ {\sqrt{{\beta _1^t}}{e^{j{\theta _1}}},{\sqrt{{\beta _2^t}}{e^{j{\theta _2}}}}, \ldots ,\sqrt{{\beta_M^t}}{e^{j{\theta _M}}}} ]^T. \\
%	 [ {\sqrt{\beta_1^r}{e^{j{\varphi _1}}},{\sqrt{\beta_2^r}{e^{j{\varphi _2}}}}, \ldots ,\sqrt{\beta_M^r}{e^{j{\varphi _M}}}} ]^T.
%\end{gathered}\right.
%	\end{aligned}
%\end{equation}
\begin{equation}
		\begin{aligned}
					\boldsymbol{v}_n
		\triangleq\left\{\begin{gathered}
				 \mathbf{v}_t=[ {\sqrt{{\beta _1^t}}{e^{j{\theta _1^t}}},{\sqrt{{\beta _2^t}}{e^{j{\theta _2^t}}}}, \ldots ,\sqrt{{\beta_M^t}}{e^{j{\theta _M^t}}}} ]^T. \\
				 \mathbf{v}_r=[ {\sqrt{\beta_1^r}{e^{j{\theta _1^r}}},{\sqrt{\beta_2^r}{e^{j{\theta_2^r}}}}, \ldots ,\sqrt{\beta_M^r}{e^{j{\theta_M^r}}}} ]^T.
			\end{gathered}\right.
			\end{aligned}
	\end{equation}
where $n\in\{t,r\}$ indicates whether the downlink signals reaching a specific user is transmitted or reflected by the STAR-RIS, $\sqrt {{\beta^t_m}},\sqrt {{\beta^r_m}}\in[0,1]$ and $\theta_m^t,\theta_m^r\in[0,2\pi),m\in\mathcal{M}=\{1,...,M\}$ denote the amplitude and phase-shift adjustments imposed on the incident signal that impinges on the $m$-th element during transmission and reflection, respectively. %Similarly, we define $\mathbf{\Theta}_n\triangleq\text{diag}(\mathbf{v}_n),n\in\{t,r\}$ as the corresponding transmission- and reflection-coefficient matrices (beamforming matrices). 
Note that due to the law of energy conservation, the sum energy of the transmitted and reflected signals has to be equal to that of the incident signals for each element, i.e. $\beta^t_m+\beta^r_m=1$ has to be satisfied\cite{STAR}.

%\begin{figure}[t!]
%  \centering
%  \includegraphics[width=3in]{eps/simple1.eps}\\
%  \caption{Illustration of an RIS aided two-user downlink transmission via simultaneous reflection and refraction.}\label{configurations}
%\end{figure}
\subsection{Transmission Model}
In this paper, we consider both OMA and NOMA transmission schemes. We assume that the direct AP-user links are blocked by obstacles and the communication is guaranteed by STAR-RIS-enabled transmission and reflection links.\footnote{This is a challenging scenario for traditional wireless networks without STAR-RISs. However, the communication requirements of users in the signal dead zone can be guaranteed with the aid of a STAR-RIS.} Then, the effective channel power gain of user $i$ on sub-channel $k$ is given by

\begin{equation}
	{\left| {{h_{k,i}}} \right|^2} = {\left| {{\mathbf{f}}_{k,i}^H{{\mathbf{\Theta}_n}}{\mathbf{g}_k}} \right|^2}=|\mathbf{v}_n^H{\mathbf{q}_{k,i}}|^2 ,\\
\end{equation}
where $\mathbf{g}_k\in\mathbb{C}^{M\times1}$ and $\mathbf{f}_{k,i}\in\mathbb{C}^{M\times1}$ denote the narrow-band quasi-static Rician fading channels from the AP to the STAR-RIS and from the STAR-RIS to user $i$, respectively. ${{\mathbf{q}}_{k,i}} = {\rm{diag}} ({\mathbf{f}}_{k,i}^H)\mathbf{g}_k$ denotes the cascaded channel vector.

\emph{1) OMA:} For OMA, a pair of users $\{i,\bar{i}\}\in\mathcal{I}_k$ on sub-channel $k$ are served by time division multiple access (TDMA) to avoid co-channel interference. Hence, the achievable communication rate of user $i$ can be expressed as
\begin{align}\label{rate3}
	R_{k,i}^O = \omega_{k,i}{\log_2}\left({1 + \frac{{{p_{k,i}}{{\left|h_{k,i} \right|}^2}}}{\omega_{k,i}{\sigma_k^2}}} \right),
\end{align}
where $p_{k,i}$ denotes the allocated transmit power of user $i$, $\sigma_k^2$ is the variance of noise, $\omega_{k,i}$ denotes time allocation factor, which satisfies $\omega_{k,i}+\omega_{k,\bar{i}}=1,\forall{k\in\mathcal{K}}$.
\begin{remark}
	\rm  We consider the general OMA case in (\ref{rate3}), where the time resource can be allocated adaptively. In conventional OMA schemes, time resources are equally allocated over each sub-channel, i.e. $\omega_{k,i}=1/2$, which serves as a special case of our considered OMA case.
\end{remark}

\emph{2) NOMA:} For power-domain NOMA, the paired users are served at the same time and frequency resource block. The superposition symbol $x_k$ to be transmitted on the $k$-th
sub-channel is
\begin{align}
	{x_k} = \sum_{i=1}^{I}\lambda_{k,i}\sqrt{p_{k,i}}s_{k,i},
\end{align}
where $s_{k,i}$ denotes the desired symbol of user $i$. 

 Then, the received signal of user $i$ on sub-channel $k$ is
\begin{align}
	y_{k,i}= |h_{k,i}|^2x_k+ z_{k,i},
\end{align}
where $z_{k,i}$ indicates the additive white Gaussian noise with zero mean and variance $\sigma_k^2$, i.e., $z_{k,i} \sim {\mathcal{C}\mathcal{N}}\left( {0,{\sigma_k^2}} \right)$.

The SIC decoding order is an essential issue for NOMA systems, where the optimal decoding order is determined by the channel gains. However, in STAR-RIS-NOMA systems, the combined channel gains can be modified by tuning the coefficients of the STAR-RIS. Therefore, the optimal decoding order will be between two different orders for each sub-channel. Denote $\pi_k(i)\in\{0,1\}$ as the decoding order of user $i$ on sub-channel $k$. For instance, for a pair of users $\{i,\bar{i}\}$ that is assigned to sub-channel $k$ with the decoding order $\pi_k(i)=1,\pi_k(\bar{i})=0$, the signal of user $\bar{i}$ is first decoded by treating user $i$'s signal as interference. Then, by removing user $\bar{i}$'s signal with
SIC techniques, user $i$ can decode its signal without the co-channel interference. In this case, $|h_{k,i}|^2\geq|h_{k,\bar{i}}|^2$ should be satisfied to guarantee that the SIC can be performed successfully. Otherwise, for the decoding order $\pi_k(i)=0,\pi_k(\bar{i})=1$, $|h_{k,i}|^2\leq|h_{k,\bar{i}}|^2$ should be satisfied.

Therefore, for any user $i$ in sub-channel $k$, the achievable communication rate can be expressed as (bit/s/Hz):
\begin{align}\label{rate1}
	R_{k,i}^N = {\log_2}\left({1 + \frac{{{p_{k,i}}{{\left|h_{k,i} \right|}^2}}}{\pi_k(\bar{i}){{p_{k,\bar{i}}}{{\left| h_{k,i} \right|}^2} +{\sigma_k^2}}}} \right),
\end{align}
where $\bar{i}$ is the index of the other user that is paired with user $i$.

\subsection{Problem Formulation}
\emph{1) Problem for OMA:} For OMA, we aim to maximize the system sum-rate by jointly optimizing channel assignment, transmit power, time allocation factor, and transmission- and reflection-coefficient vectors, subject to the rate requirements. Hence, the optimization problem can be formulated as follows:

\begin{subequations}\label{pro_oma}
	\begin{align}
		\mathop {\max}\limits_{\left\{\boldsymbol{\lambda},\mathbf{p} ,\boldsymbol{\omega},\mathbf{v}_n\right\}} &\;\;\sum_{k=1}^{K}\sum_{i=1}^{I}\lambda_{k,i}R_{k,i}^O  \\
		\label{QoS}{\rm{s.t.}}\;\;&R_{k,i} \ge {\gamma},\forall k \in {\mathcal{K}},\forall{i\in\mathcal{I}},\\
		\label{energy}&\beta^t_m+\beta^r_m=1,\forall m \in {\mathcal{M}},\\
		\label{unitmo}&|\theta_m^n|=1,\forall m \in {\mathcal{M}},n \in \{t,r\},\\
		\label{positvepower}&p_{k,i}\geq{0}, \forall k \in {\mathcal{K}},\forall{i\in\mathcal{I}},\\
		\label{totalpower}&\sum_{k=1}^K\sum_{i=1}^I{p_{k,i}}\leq{P_{\text{max}}}, \\
		\label{binary}&\lambda_{k,i}\in\{0,1\},\forall k,i,\\
		%\label{max_c}&|\mathcal{I}_k|\leq{N_{\text{max}}}, \forall k,  \\
		\label{max_c}&1\leq\sum_{i=1}^I\lambda_{k,i}\leq{2},\forall k\in\mathcal{K},\\
		\label{band}&\omega_{k,i}+\omega_{k,\bar{i}}=1,\forall{k\in\mathcal{K}},
	\end{align}
\end{subequations}
where $\boldsymbol{\lambda}=\{\lambda_{1,1},..,\lambda_{K,I}\}$ is the channel assignment vector, $\mathbf{p}=\{p_{1,1},...,p_{K,I}\}$ is the power allcoation vector, and $\boldsymbol{\omega}=\{\omega_{1,1},...,\omega_{K,I}\}$ denotes the time allocation vector.
Constraint (\ref{QoS}) guarantees the communication rate of each user should be no lower than $\gamma$. 
Constraint (\ref{energy}) represents the law of energy conservation. (\ref{unitmo}) is the unimodular phase shift constraint for each element of the STAR-RIS. Constraint (\ref{positvepower}) and (\ref{totalpower}) ensures that the allocated power for each user is positive and the total budget is $P_{\text{max}}$, respectively. Constraints (\ref{binary}) and (\ref{max_c}) indicate that the channel assignment variable is binary and the maximum number of users that are multiplexed on each subchannel is two, respectively.  (\ref{band}) is the time allocation factor constraint.

\emph{2) Problem for NOMA:} For NOMA, the joint optimization problem of channel assignment, decoding order, transmission- and reflection-coefficient vectors, and power allocation can be formulated as

\begin{subequations}\label{problem}
\begin{align}
\mathop {\max}\limits_{\left\{\boldsymbol{\lambda},\boldsymbol{\pi},\mathbf{v}_n,\mathbf{p}\right\}} &\;\;\sum_{k=1}^{K}\sum_{i=1}^{I}\lambda_{k,i}R_{k,i}^N  \\
{\rm{s.t.}}\;\;
\label{sic}&\left\{\begin{gathered}
	\text{if}\ \pi_k(i)=1,\pi_k(\bar{i})=0,|h_{k,i}|^2\geq|h_{k,\bar{i}}|^2, \\
	 \text{if}\ \pi_k(i)=0,\pi_k(\bar{i})=1,|h_{k,i}|^2\leq|h_{k,\bar{i}}|^2.
\end{gathered}\right.\\
&\rm(\ref{QoS})-(\ref{max_c}).
\end{align}
\end{subequations}
where $\boldsymbol{\pi}=\{\pi_1(1),...,\pi_K(I)\}$ is the decoding order vector. Constraint (\ref{sic}) ensures the SIC decoding can be conducted successfully.

\begin{theorem}
	Problem (\ref{pro_oma}) and (\ref{problem}) are non-deterministic polynomial- (NP-) hard even when only the channel assignment is considered.
\end{theorem}
\begin{proof}\renewcommand{\qedsymbol}{}
	\emph{See Appendix A.}
\end{proof}

The main challenges to solve problem (\ref{pro_oma}) and (\ref{problem}) are as follows: Firstly, due to the co-existence of integer and continuous variables, as well as the non-linear objective function and constraints, the formulated sum-rate maximization problem is an MINLP problem, which is non-trivial using traditional optimization methods. Even if only the channel assignment problem is considered, it is proved to be NP-hard. Secondly, the optimization variables of transmit power and beamforming are coupled, which makes the problem non-convex. Finally, for NOMA problem, since the STAR-RIS can control the combined channel, it is difficult to obtain the optimal decoding order.
\subsection{Algorithm Overview}

%\begin{algorithm}[t]
%	\caption{Alternative Algorithm for multi-carrier STAR-RIS aided NOMA networks}
%	\label{alg1}
%	\begin{algorithmic}[1]
%		\STATE \textbf{Initialize} Randomly initialize the channel assignment scheme $\boldsymbol{\lambda}$. $l=0$. 
%		\STATE \textbf{repeat}
%		\STATE \qquad Jointly optimizing beamforming matrix, decoding order, and power allocation via \textbf{Algorithm 2}. 
%		
%		\STATE \qquad Optimize the channel assignment using \textbf{Algorithm 3} or LSA.
%		\STATE \qquad $l=l+1$;
%		\STATE \textbf{until} The objective function value converges or $l=l^{\text{max}}$.
%	\end{algorithmic}
%\end{algorithm}
To overcome the challenges of the formulated problem (\ref{pro_oma}) and (\ref{problem}), we propose efficient algorithms for OMA and NOMA, %given in sections IV and V, 
respectively. For an overview, the roadmap of proposed algorithms is summarized in Fig. 2. 
\begin{figure}[t!]
	\centering
	\includegraphics[width=1\textwidth]{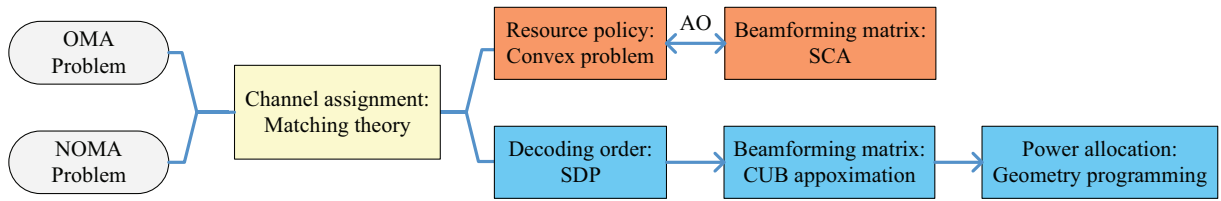}\\
	\caption{Roadmap of proposed algorithms for solving problem OMA and NOMA.}\label{roadmap}
\end{figure}
First, we propose efficient algorithms utilizing matching theory to deal with the NP-hard sub-channel assignment problem. Particularly, for OMA, we consider a general design, while for NOMA, we consider transmission-reflection user-pairing to reduce the complexity. % which is summarized in \textbf{Algorithm 1}. 
After channel assignment, an alternating optimization-based algorithm is proposed to solve the resulting OMA problem. %, which is presented in \textbf{Algorithm 2}. 
Specifically, the first subproblem of power and time allocation is proved to be convex. Then, the second subproblem of transmission- and reflection-coefficient vector design is solved using SCA. Finally, the two subproblems are optimized iteratively until convergence. For NOMA, a three-step approach is proposed after channel assignment. % which is summarized in  $\textbf{Algorithm 4}$.
Specifically, the decoding orders are first designed by maximizing the overall channel gains using semidefinite programming (SDP). Next, the transmission- and reflection-coefficient vector is optimized using CUB approximation. % which is summarized in $\textbf{Algorithm 3}$.
Finally, the power allocation sub-problem is transformed to a GP, which is convex.

\section{Proposed Solutions For OMA}
In this section, we first utilize matching theory to design the channel assignment scheme. Then, an AO-based algorithm is proposed to solve the resulting OMA problem. 

\subsection{Channel Assignment Scheme for OMA}
In this subsection, we focus on the channel assignment strategy, whose problem can be formulated as follows:
\begin{subequations}\label{p8}
	\begin{align}
		\mathop {\max}\limits_{\left\{\boldsymbol{\lambda} \right\}} &\;\;\sum_{k=1}^{K}\sum_{i=1}^{I}\lambda_{k,i}R_{k,i}^O  \\
		{\rm{s.t.}}\;\;&\rm(\ref{binary}),(\ref{max_c}).
	\end{align}
\end{subequations}

%\subsection{Exploration and Swap Based Algorithm} 
To obtain the globally optimal solution of problem (\ref{p8}), a straightforward method is to exhaustively search over all the combinations between sub-channels and users, which is computationally unacceptable. To solve this difficulty, we reformulate the problem as a matching game between sub-channels and users and solve it by a
low-complexity matching theory-based algorithm. %Despite the different expressions of achievable data rates, the matching theory is shown to be applicable to both OMA and NOMA. 
To describe our proposed algorithm well, we first introduce some basic concepts of the matching theory.
%combination

\emph{Definition 1 (Two-sided Matching):} Consider two disjoint sets, the subchannel set $\mathcal{K}=\{1,...,K\}$ and the user set $\mathcal{I}=\{1,...,I\}$. A many-to-one (two-to-one) matching $\Phi$ is defined as a function from $\mathcal{K}\cup\mathcal{I}$ to all subsets of $\mathcal{K}\cup\mathcal{I}$ such that
\begin{equation*}
	\begin{aligned}
		&1) \ \ \Phi(k)\subseteq\mathcal{I},\Phi(i)\in\mathcal{K};\\
		&2) \ \ \left|\Phi(k)\right|={2},\left|\Phi(i)\right|=1;\\
		&3) \ \ k=\Phi(i)\Leftrightarrow{i\in\Phi(k)}.
	\end{aligned}
\end{equation*}

Condition 1) implies that each subchannel is matched with a subset of users and each user is matched with one of the subcarriers. Condition 2) states that each user can occupy only one subchannel and each subcarrier can be assigned to two users. Condition 3) implies that channels and users are matched with each other. %Note that when $M_{\text{max}}=1$, it becomes a definition of many-to-one matching.

\emph{Definition 2 (Utility Function):} Given a matching state $\Phi$, the utility functions of a user $i$ and a subchannel $k$ are defined as: $U_i(\Phi)=R_{k,i}$, $U_k(\Phi)=\sum_{i\in\Phi(k)}R_{k,i}$.

\emph{Definition 3 (Swap Matching):} For users $i,\widetilde{i}$ and subchannels $k,\widetilde{k}$ with matching state $\Phi(i)=k,\Phi(\widetilde{i})=\widetilde{k}$, a swap matching is
\begin{equation*}
	\Phi_i^{\widetilde{i}}=\left\lbrace \Phi\setminus\left\{(i,k),(\widetilde{i},\widetilde{k})\right\}\cup\left\{(i,\widetilde{k}),(\widetilde{i},k)\right\}\right\rbrace 
\end{equation*}

One swap matching process enables two users to switch their assigned subchannels. Based on swap matching, we define a swap-blocking pair.

\emph{Definition 4 (Swap-blocking pair):} Given a matching state $\Phi$ with $\Phi(i)=k,\Phi(\widetilde{i})=\widetilde{k}$, a pair of users $(i,\widetilde{i})$ is a swap-blocking pair if and only if
\begin{equation*}
	\begin{aligned}
		&1)\ \ \forall{x}\in\{i,\widetilde{i},k,\widetilde{k}\},U_x(\Phi_i^{\widetilde{i}})\geq{U_x(\Phi)}\\
		&2) \ \ \exists{x}\in\{i,\widetilde{i},k,\widetilde{k}\},U_x(\Phi_i^{\widetilde{i}})>{U_x(\Phi)}
	\end{aligned}
\end{equation*}
where $x$ denotes any of the players involved (either user or sub-channel). Based on the above definitions, the proposal of two users to switch their channels will be approved only if the following conditions are met: Condition 1) indicates that the utility of any of the players should not be reduced after a swap operation. Condition 2) implies that at least one of the players' rates will increase.

%Motivated to address this issue, we take `exploration' into account and develop a novel exploration and swap-based algorithm (ESA) to approach the optimum solution. More specifically, in ESA, we can initialize the matching state $\Phi$ randomly. Then, in the swapping phase, we keep searching for two users $(i,\widetilde{i})$ to swap their channels if the following conditions are satisfied: 1) If $(i,\widetilde{i})$ is a swap-blocking pair, the swap will be executed immediately. 2) If it is not, $\Phi_i^{\widetilde{i}}$ will be executed with probability
%\begin{equation} \label{prob}
	%	p_r=\frac{1}{1+e^{-s\frac{U_{\text{total}}(\Phi_i^{\widetilde{i}})-U_{\text{total}}(\Phi)}{U_{\text{total}}}}},
	%\end{equation}
%when $U_{\text{total}}(\Phi_i^{\widetilde{i}}){>}U_{\text{total}}(\Phi)$, where $U_{\text{total}}$ denotes the total utility of involved players under matching $\Phi$, $s$ is a constant parameter. The above swapping rule facilitates AP to find potential user pairs to swap their channels each if they are not swap-blocking. As a result, the system sum-rate may benefit from this exploration. Moreover, the algorithm no longer depends on the initialization scheme. The detail of proposed ESA is summarized in \textbf{Algorithm} 3.
%\subsection{Low-complexity Location-based Swapping Algorithm}
%The ESA may bring near-optimal behavior as the exhaustive search-based algorithm but is of high complexity as there are more users and subchannels. 

Next, we introduce our channel assignment design for OMA. In the initialization phase, we assume that the time allocation factor $\omega_{k,i}$ for each user is 0.5 and the amplitude adjustment of each element $\beta_m^t=\beta_m^r=0.5,\forall{m\in\mathcal{M}}$. Then, each side of the STAR-RIS can be regarded as an independent conventional RIS and the phase-shifts of each element can be initialized by the traditional method proposed in\cite{ofdm}. Then, all users calculate their channel gains and propose to join the channel that provides the highest gain. At the AP, two users' proposals are approved on each sub-channel and the matching state is initialized as $\Phi$.

We can use matching theory to further improve the channel gain of each user. Based on our definitions, we can formulate a \emph{swap operation-enabled algorithm} as many previous works have considered\cite{nomaris3,power2}, which is briefly described as follows: The AP keeps searching for two users that form a swap-blocking pair and execute the swap operation. The process continues until the final matching state $\Phi'$ does not contain any swap-blocking pair. 

%Note that the above process requires solving a SDP problem\cite{ofdm} and executing swap operation among all users, which is computationally costly. In the next section, we will propose a low-complexity channel assignment scheme for NOMA, which is also shown to be applicable to OMA.

\subsection{Proposed AO-Based Algorithm}
In this subsection, we focus on the joint optimization of transmit power, time allocation factor, and transmission- and reflection-coefficient vectors to solve the resulting OMA problem after channel assignment. To deal with the highly coupled variables, we propose an alternating optimization algorithm, which iteratively optimizes two groups of variables, i.e., $\{\mathbf{p},\boldsymbol{\omega}\}$ and $\{\mathbf{v}_n\}$, respectively.

\emph{Optimization with fixed beamforming vectors:} First, we optimize the transmit power $\mathbf{p}$ and time block $\boldsymbol{\omega}$ allocation with fixed beamforming vector $\mathbf{v}_n$. The subproblem is written as:
\begin{subequations}\label{pro_oma2}
	\begin{align}
		\mathop {\max}\limits_{\left\{\mathbf{p} ,\boldsymbol{\omega}\right\}} &\;\;\sum_{k=1}^{K}\sum_{i\in\mathcal{I}_k}R_{k,i}^O  \\
		{\rm{s.t.}}\;\;&\rm(\ref{QoS}),(\ref{positvepower}),(\ref{totalpower}),(\ref{band}).
	\end{align}
\end{subequations}

It can be verified that problem (\ref{pro_oma2}) is a standard convex optimization problem since the function $f(x,y)=x\log_2(y/x)$ is jointly concave with respect to $x$ and $y$. As such, (\ref{pro_oma2}) can be efficiently solved via existing solvers, e.g., CVX\cite{cvx}.

\begin{remark}
	\rm For conventional OMA transmission scheme with $\omega_{k,i}=0.5$, the power allocation scheme can be given by the well-known water-filling solution, which can be expressed as close-formed\cite{wirelesscom}.
\end{remark}

\emph{Optimization with fixed power and time allocation:} Next, with fixed power and time allocation, the subproblem of optimizing transmission- and reflection- coefficient vector is formulated as
\begin{subequations}\label{pro_oma_matrix}
	\begin{align}
		\mathop {\max}\limits_{\left\{\mathbf{v}_n\right\}} &\;\;\sum_{k=1}^{K}\sum_{i\in\mathcal{I}_k}R_{k,i}^O  \\
		{\rm{s.t.}}\;\;&\rm(\ref{QoS}),(\ref{energy}),(\ref{unitmo}).
	\end{align}
\end{subequations}

Since the objective function of problem (\ref{pro_oma_matrix}) is non-concave, we first introduce a set of auxiliary variables $a_{k,i},b_{k,i},\chi_{k,i}$, and the problem is equivalently transformed to  
\begin{subequations}
	\begin{align}
		\mathop {\max}\limits_{\left\{a_{k,i},b_{k,i},\chi_{k,i},\mathbf{v}_n\right\}}& \omega_{k,i}{\log_2}\left({1 + \frac{{{P_{k,i}}{{\chi_{k,i}}}}}{\omega_{k,i}{\sigma_k^2}}} \right)\\	
		\label{con_real}{\rm{s.t.}}\;\;&a_{k,i}=\text{real}(\mathbf{v}_n^H{\mathbf{q}_{k,i}}),	\\
		\label{con_imag}&b_{k,i}=\text{imag}(\mathbf{v}_n^H{\mathbf{q}_{k,i}}),\\
		\label{sca_con}&\chi_{k,i}^2\leq{a_{k,i}^2+b_{k,i}^2},\\
		\label{time}&\sum_{i\in\mathcal{I}_k}\omega_{k,i}=1,\forall{k\in\mathcal{K}}.
	\end{align}
\end{subequations}

Since constraint (\ref{sca_con}) is still non-convex, we invoke successive convex approximation (SCA) technique to deal with it. Specifically, we use first-order Taylor expansion as an upper bound to represent the right-hand-side of (\ref{sca_con}), i.e.,
\begin{align}\label{sca_con2}
	\chi_{k,i}\leq\bar{a}_{k,i}^2+\bar{b}_{k,i}^2+2\bar{a}_{k,i}(a_{k,i}-\bar{a}_{k,i}^2)+2\bar{b}_{k,i}(b_{k,i}-\bar{b}_{k,i}^2),
\end{align}
where $\bar{a}_{k,i}$ and $\bar{b}_{k,i}$ are stationary points that can be updated by the solutions of the last iteration. 

Finally, we consider the following optimization problem
\begin{subequations}\label{pro_oma3}
	\begin{align}
		\mathop {\max}\limits_{\left\{a_{k,i},b_{k,i},\chi_{k,i},\mathbf{v}_n\right\}}& \omega_{k,i}{\log_2}\left({1 + \frac{{{P_{k,i}}{{\chi_{k,i}}}}}{\omega_{k,i}{\sigma_k^2}}} \right)\\	
		{\rm{s.t.}}\;\;&\rm (\ref{con_real}),(\ref{con_imag}),(\ref{time}),(\ref{sca_con2}).
	\end{align}
\end{subequations}
which is a convex optimization problem. Thus, it can be solved efficiently via standard solvers like CVX\cite{cvx}. The procedure of AO-based method for solving OMA problem is given in \textbf{Algorithm \ref{alg_oma_ao}}.

\begin{algorithm}[t]
	\caption{Alternating Optimization Algorithm for Solving OMA Problem}
	\label{alg_oma_ao}
	\begin{algorithmic}[1]
		\STATE \textbf{Initialize} Initialize the channel assignment $\lambda_{k,i}$, tolerance $\epsilon$. Set $t=0$. 
		\STATE \textbf{repeat:} 
		\STATE \qquad $t=t+1$;
		\STATE \qquad Solve problem (\ref{pro_oma2}) to obtain the optimal power and time allocation $p^t_{k,i},\omega_{k,i}^t$.
		\STATE \qquad Optimize the transmission and reflection beamforming vector by solving (\ref{pro_oma3}).
		\STATE \textbf{until} The increase of the objective function is within $\epsilon$.
		\STATE \textbf{Output} Converged solution $\mathbf{p}^*,\boldsymbol{\omega}^*,\mathbf{v}_n^*$.
	\end{algorithmic}
\end{algorithm}

\subsection{Stability, Convergence and Complexity Analysis}
\begin{proposition} \label{proposition2}
	Following the swap operation-enabled algorithm, one match will be two-sided exchange stable but may not converge to a globally optimal solution. 
	
\end{proposition}
\begin{proof}
	The swap operation enabled algorithm will keep searching until there is no swap-blocking pair. Thus, when it reaches the final state, none of the players' utility will increase without affecting others'. i.e., the final matching state is \emph{two-sided exchange stable}. Since all players benefit from each swap operation, the algorithm is guaranteed to converge after a finite number of steps. However, the algorithm may not reach the global optimum. For example, assume that the optimal matching is $\{(i,k),(\widetilde{i},\widetilde{k})\}$ and the current matching is $\{(i,\widetilde{k}),(\widetilde{i},k)\}$. If user $i$ proposes to swap its channel, its rate may increase a lot but user $\widetilde{i}$'s rate may decrease. Thus, $(i,\widetilde{i})$ can not form a swap-blocking pair and the swap operation may never be executed. In this case, the algorithm converges to a local optimum.
\end{proof}

From \textbf{Proposition} \ref{proposition2}, we find that the swap operation-enabled algorithm may not lead to global optimum since it is significantly affected by the initial matching state. Thus, it is essential to consider the initialization method as well as the swap-operation scheme carefully. Note that invoking AO, the objective values of problem (\ref{pro_oma2}) and (\ref{pro_oma3}) are non-decreasing during each iteration\cite{ofdm}. Moreover, since the system throughput is upper-bounded by a finite value, the proposed AO algorithm is guaranteed to converge.

Next, we analyze the computational complexity of the proposed algorithms. For swap operation-enabled algorithm for channel assignment, each active user searches for other users that are associated with different channels to perform potential swap operations. Then, it requires $\frac{I(I-1)}{2}$ steps of searching, leading to at most the same number of swap operations. %Denote the iteration numbers of the overall algorithm as $I_{\text{itr2}}$, the complexity is then calculated as $\mathcal{O}(I_\text{itr2}I^2)$. 
 The complexity of solving problem (\ref{pro_oma2}) is given by $\mathcal{O}(I^{3.5})$, while the complexity for solving problem (\ref{pro_oma3}) is given by $\mathcal{O}(M^{4.5}I^3)$. 
\section{Proposed Solutions for NOMA}
In this section, we first propose a low-complexity channel assignment sceme for NOMA. However, the resulting resource allocation problem is still challenging due to the binary decoding orders and the co-channel interference item in (\ref{rate1}) compared with OMA. To facilitate solving it, we propose an efficient three-stage algorithm to optimize the decoding order, power allocation and beamforming vectors step-by-step.
%1) LSA-1: In this scheme, we always pair a T user with an R user in the initialization phase and each user only searches for other users that belong to same region to form a swap-blocking pair. Specifically, a user $i\in\mathcal{I}_t$ searches for all users $j\in\mathcal{I}_t$ such that $\Phi(j)\neq\Phi(i)$, a user $\widetilde{i}\in\mathcal{I}_r$ searches for all users $\widetilde{j}\in\mathcal{I}_r$ such that $\Phi(\widetilde{j}j\neq\Phi(\widetilde{i})$
%
%
%2) LSA-2: In this scheme, we primarily pair users of the same side and execute swap operations among users that belong to the same region but transmit over different subchannels.
%%For example, a T user $i\in\mathcal{I}_t$ searches for all users $j\in\mathcal{I}_r$ such that $\Phi(j)\neq\Phi(i)$, and vice versa.
\subsection{Channel Assignment Scheme for NOMA}
The channel assignment problem for NOMA is given by

\begin{subequations}\label{channel_assignment_noma}
	\begin{align}
		\mathop {\max}\limits_{\left\{\boldsymbol{\lambda} \right\}} &\;\;\sum_{k=1}^{K}\sum_{i=1}^{I}\lambda_{k,i}R_{k,i}^N  \\
		{\rm{s.t.}}\;\;&\rm(\ref{binary}),(\ref{max_c}).
	\end{align}
\end{subequations}

Based on the definitions in section III, the formulated channel assignment problem for NOMA is a two-to-one matching with peer effects, also known as externalities\cite{matching1}. More specifically, due to interference item in (\ref{rate1}), the utility of each user depends on the subset of users that share the same subcarrier. Thus, an individual user cares about the allocated channel and other users sharing the same channel. Similarly, the utility of a subchannel is related to all users who have inner-relationship through multiplexing. Due to the peer effects, the utility of players keeps changing over the matching process, which makes the matching mechanisms complex to design. To deal with it, we will propose efficient initialization and swap operation scheme to achieve exchange stability.

Inspired by the unique topological characteristic of STAR-RIS that it divides the whole region into two parts, we propose a low-complexity location-based matching algorithm. LMA performs user-pairing and swap operations according to the regions that the users are located in. Moreover, due to the new DoFs of passive beamforming design that STAR-RIS brings, the channel disparities of users from different regions may be enlarged. It is intuitively reasonable to form a \emph{transmission-reflection user-pairing} and swap-matching scheme for NOMA. The detailed processes of the proposed LMA are as follows:

\textbf{1) Initialization Phase:} 
In LMA, we always pair a T user with an R user on a sub-channel in the initialization phase. Since the `stronger' users usually achieve higher gains from the resource allocation scheme, we can initialize the transmission- and reflection-coefficient vectors as
\begin{equation}\label{ini_phase}
	\beta_m^n=0.5,\theta_m^n=-\text{angle}(v_n^*),n\in\{t,r\},m\in\mathcal{M}
\end{equation}
where $v_n^*$ is the complex phase shift that leads to the maximum cascaded channel gain of a T user or an R user. Specifically, $v_n^*$ is the element of the optimal transmission or reflection coefficient vectors $\textbf{v}_n^*$, which are calculated as $\textbf{v}_n^*=\arg\underset{\textbf{v}_n}{\max}|{{\mathbf{v}}_n^H}{{\mathbf{q}}_{k,i}}|^2,n\in\{t,r\}$. In short, we calculate each user $i$'s potential maximum cascaded channel gain $\sum_{m=1}^M|q_{k,i}^m|^2$ and pick the highest one among T users and R users, respectively. Then we adjust the phase-shifts of the STAR-RIS to achieve the picked channel gain for transmission and reflection beamforming.
\begin{algorithm}[t]\label{lma}
	\caption{LMA for Channel Assignment}
	\label{alg1}
	\begin{algorithmic}[1]
		\STATE \textbf{Initialization Phase:} 
		\STATE \qquad Initialize the transmission and reflection coefficient vectors $\textbf{v}_n^*$ according to (\ref{ini_phase}). Set the power allocation $p_{k,i}=P_{\text{max}}/I$.
		\STATE \qquad All users calculate their equivalent channel gains $h_{k,i}=|{\mathbf{q}_{k,i}^H}\mathbf{v}_n^*|^2$.
		\STATE \textbf{repeat}
		\STATE \qquad Each user proposes to join the channel that provides the highest gain.
		\STATE \qquad Each channel accepts one proposal from T users and R users, respectively, and rejects others. 
		\STATE \textbf{until} There are no un-matched users.
		\STATE \textbf{Swapping Phase:} 
		\STATE \textbf{repeat}
		\STATE \qquad For any user $i\in\mathcal{T}$, it searches for all users $j\in\mathcal{T}$ such that $\Phi(j)\neq\Phi(i)$
		\STATE \qquad \textbf{if}  {$(i,j)$ is a swap-blocking pair }\textbf{then}	
		%\IF{blabla}
		\STATE \qquad\qquad Update $\Phi=\Phi_i^{j}$.
		\STATE \qquad \textbf{else} 
		%\ELSE 
		\STATE \qquad \qquad Keep the current matching state.
		\STATE \qquad \textbf{end if}
		%\ENDIF
		\STATE \textbf{until} All the users have been searched.
		\STATE \textbf{Output} Channel assignment vector $\boldsymbol{\lambda}^*$.
	\end{algorithmic}
\end{algorithm}
Next, based on the above transmission- and reflection-coefficient vectors, the users initialize their matching state as follows: Firstly, each user calculates channel gain according to (\ref{ini_phase}). Secondly, each users proposes to join the channel that provides the highest gain and has never rejected it. Thirdly, to ensure that a T user and an R user are paired, each channel only accepts one proposal from the T user and R user with the highest channel gain, respectively, and rejects others. Repeat the second and third steps until there are no unmatched users.

\textbf{2) Swapping Phase}: To further improve the effectiveness of the algorithm, swap operations are executed among users. In particular, each user only searches for other users that belong to the same region to form a swap-blocking pair, i.e. a user $i\in\mathcal{T}$ searches for all users $j\in\mathcal{T}$ such that $\Phi(j)\neq\Phi(i)$. %a user $\widetilde{i}\in\mathcal{I}_r$ searches for all users $\widetilde{j}\in\mathcal{I}_r$ such that $\Phi(\widetilde{j})\neq\Phi(\widetilde{i})$. 
If two users are swap-blocking, they swap their channels immediately. Repeat the above procedure until there are no swap-blocking pairs.

The details of the proposed LMA is summarized in \textbf{Algorithm 2}.
\begin{remark}
	\rm The proposed LMA can also be exploited to OMA transmission scheme, even though there are no peer effects. Since the swap operations between user $i,j$ will not affect the utilities of other users that are assigned to the same sub-channels, user $i,j$ can always swap their channels if they are swap-blocking. Thus, the final matching state will be two-sided exchange stable.
\end{remark}
\subsection{A Low Complexity Scheme for Decoding Order Determination}
After channel assignment, for the considered STAR-RIS-NOMA communication system, it is essential to decide the decoding order properly since the transmission and reflection coefficient may affect the optimal decoding order. For the whole system, the optimal decoding order will be any of the $2^K$ different orders and the computational complexity to solve the original problem will grow exponentially with more users, which is prohibitive. Therefore, we propose a low-complexity algorithm to determine the decoding orders by finding beamforming vectors that maximize the sum of combined channel gains. The formulated optimization problem is as follow:
\begin{subequations}\label{pro_max_channel_gain}
	\begin{align}
		\mathop {\max}\limits_{\left\{\mathbf{v}_n \right\}} &\;\;\sum_{k=1}^{K}\sum_{i\in\mathcal{I}_k}|h_{k,i}|^2\\
		\label{energy_con}{\rm{s.t.}}\;\;&\beta^t_m+\beta^r_m=1,\forall m \in {\mathcal{M}},\\
		\label{unitmo2}&|\theta_m^n|=1,\forall m \in {\mathcal{M}},n \in \{t,r\},
	\end{align}
\end{subequations}

In problem (\ref{pro_max_channel_gain}), the quadratic form ${\left| {{h_{k,i}}} \right|^2}=|\mathbf{v}_n^H{\mathbf{q}_{k,i}}|^2$ can be rewritten as $\mathbf{v}_n^H{\mathbf{Q}_{k,i}}\mathbf{v}_n$, where ${\mathbf{Q}_{k,i}}={\mathbf{q}_{k,i}}{\mathbf{q}_{k,i}^H}$. By introducing slack matrixs  $\mathbf{V}_n=\mathbf{v}_n\mathbf{v}_n^H,n\in\{t,r\}$, which are rank-one and positive semidefinite (PSD),
we have ${\left| {{h_{k,i}}} \right|^2}={\rm{Tr}}(\mathbf{Q}_{k,i}\mathbf{V}_n)$. Thus, problem (\ref{pro_max_channel_gain}) can be transformed to 
\begin{subequations}\label{pro_max_channel_gain2}
	\begin{align}
		\mathop {\max}\limits_{\left\{\mathbf{V}_n \right\}} &\;\;\sum_{k=1}^{K}\sum_{i\in\mathcal{I}_k}{\rm{Tr}}(\mathbf{Q}_{k,i}\mathbf{V}_n)\\
		\label{energy_con2}{\rm{s.t.}}\;\;&{\rm{Diag}}(\mathbf{V}_t)+{\rm{Diag}}(\mathbf{V}_r)=\textbf{1}^M,\\
		\label{psd}&\mathbf{V}_n\succeq0,n \in \{t,r\},\\
		\label{rank1_2}&{\rm{rank}}(\mathbf{V}_n)=1,n \in \{t,r\},.
	\end{align}
\end{subequations}

Note that the constraint (\ref{energy_con}) and (\ref{unitmo2}) are equivalently expressed as (\ref{energy_con2}), where $\textbf{1}^M$ is a length-$M$ vector with all the elements is one. Problem (\ref{pro_max_channel_gain2}) is a standard SDP problem, which can be solved by existing solvers in CVX\cite{cvx}. The only challenge lies in the rank-one constraint (\ref{rank1_2}). To deal with it, we can apply the Gaussian randomization to obtain a sub-optimal solution\cite{sdr}. 
The decoding orders are then determined by comparing the effective channel gains of the pairs ob each sub-channel. 

After the channel assignment and decoding orders are settled, in the following, we further denote $(k,(1))$ and $(k,(2))$ as the users that have decoding order 1 and 0 on sub-channel $k$, respectively. Then, with given decoding orders, constraint (\ref{sic}) for successful SIC decoding is equvalently transformed to 
\begin{equation}\label{sic2}
	{\rm{Tr}}(\mathbf{Q}_{k,(1)}\mathbf{V}_n)\geq{\rm{Tr}}(\mathbf{Q}_{k,(2)}\mathbf{V}_n),\ n\in\{t,r\}.
\end{equation}
which is a linear constraint.

%\begin{remark}
%\rm \emph{Maximum cascaded channel-based decoding order design:} When the number of users is small, we can simply decide the decoding orders by the cascaded channel gains. Specifically, denote $c_{k,i}=\max|{{\mathbf{v}}_k^H}{{\mathbf{q}}_{k,i}}|^2=\sum_{m=1}^M|q_{k,i}^m|^2=\Vert\mathbf{q}_{k,i}\Vert_1^2$ as the maximum cascaded channel gain of user $i$, where $q_{k,i}^m$ is the $m$-th element of $\mathbf{q}_{k,i}$. The decoding order is then decided by the ascending order of $c_{k,i}$. The cascaded channel-based algorithm does not require solving an SDP problem, which can be applied to  small-scale networks. 
%\end{remark}

\subsection{Design of Transmission- and Reflection-Coefficient Vector}

With determined NOMA decoding order and assuming that the power allocation is fixed, the optimization problem of transmission- and reflection beamforming vector is formulated as:

\begin{subequations}\label{p3}
 	\begin{align}
 		\mathop {\max}\limits_{\left\{\mathbf{v}_n\right\}} &\;\;\sum_{k=1}^{K}\sum_{i=1}^{2}R_{k,(i)}^N  \\
 		{\rm{s.t.}}\;\;&\rm(\ref{QoS})-(\ref{unitmo}).
 	\end{align}
\end{subequations}

Since the objective function of problem (\ref{p3}) is not concave, we first transform it to 
\begin{subequations}\label{p4}
	\begin{align}
		\mathop {\max}\limits_{\left\{\mathbf{v}_n,\boldsymbol{\zeta} \right\}} \label{obj4}&\;\;\sum_{k=1}^{K}\sum_{i=1}^{2}\log_2({1+\zeta_{k,(i)}})  \\
		\label{slack}{\rm{s.t.}}\;\;&{\log_2}\left({1 + \frac{{{p_{k,(2)}}{{\left|h_{k,(2)} \right|}^2}}}{{{p_{k,(1)}}{{\left| h_{k,(2)} \right|}^2} +{\sigma_k^2}}}} \right) \geq\zeta_{k,(2)},\log_2\left(1+\frac{p_{k,(1)}|h_{k,(1)}|^2}{\sigma_k^2}\right)\geq\zeta_{k,(1)},\\
		&\rm(\ref{QoS})-(\ref{unitmo}),
	\end{align}
\end{subequations}
where we introduce auxiliary set $\boldsymbol{\zeta}=\{\zeta_{1,(1)},...,\zeta_{K,(2)}\}$ as a new variable, whose elements are equal to the left-hand side of each inequality in (\ref{slack}).
\begin{proposition}
	Problem (\ref{p4}) is equivalent to problem (\ref{p3}) with given decoding order and power allocation.
\end{proposition}
\begin{proof}%\renewcommand{\qedsymbol}{}
%\emph{See Appendix A.}	
If any of the constraints in (\ref{slack}) is satisfied with inequality, then we can always increase the corresponding $\zeta_{k,(i)}$ to make (\ref{slack}) meet with equality. And since the function $\log_2({1+\zeta_{k,(i)}})$ increases monotonically with $\zeta_{k,(i)}$, the objective function (\ref{obj4}) will reach its optimal value when (\ref{slack}) holds with equality. Therefore, problem (\ref{p3}) and problem (\ref{p4}) have the same optimal solution.
\end{proof}
With the introduced auxiliary variables, the objective funtion (\ref{obj4}) is concave and constraint (\ref{QoS}) is convex. However, constraint (\ref{slack}) is still non-convex. To tackle it, we first rewrite constraint (\ref{slack}) as follows:
\begin{equation}\label{re1}
	 p_{k,(1)} {{\left| h_{k,(2)} \right|}^2}\zeta_{k,(2)} +{\sigma_k^2}\zeta_{k,(2)}\leq{p_{k,(2)}}{\left|h_{k,(2)} \right|^2},\ {\sigma_k^2}\zeta_{k,(1)}\leq{p_{k,(1)}}{\left|h_{k,(1)} \right|^2}.
\end{equation}

Following the similar procedure of dealing with the quadratical form $|h_{k,(i)}|^2$ as that in the last subsection, constraint (\ref{re1}) is equivalent to:
%\begin{equations}
	\begin{align}
	&{\sigma_k^2}\zeta_{k,(1)}\leq{p_{k,(1)}}{{\rm{Tr}}(\mathbf{Q}_{k,(1)}\mathbf{V}_n)},\\
	\label{re2}& p_{k,(1)} {{\rm{Tr}}(\mathbf{Q}_{k,(2)}\mathbf{V}_n)}\zeta_{k,(2)} +{\sigma_k^2}\zeta_{k,(2)}\leq{p_{k,(2)}}{{\rm{Tr}}(\mathbf{Q}_{k,(2)}\mathbf{V}_n)},\\
	&\rm{(\ref{psd}),(\ref{rank1_2})}.
	\end{align}
%\end{equations}

Constraint (\ref{re2}) is still non-convex due to the first part of the left-hand side.
Inspired by the method in \cite{cub}, we use a convex upper bound (CUB) to approximate this non-convex term. Specifically, note that for $x,y>0$, and a positive constant $\alpha$, $f(x,y)=xy$ is always upper bounded by $g(x,y)=\frac{\alpha}{2}x^2+\frac{1}{2\alpha}y^2$, i.e., $f(x,y)\leq{g(x,y)}$. Thus, we have:
\begin{equation}\label{re3}
	{{\rm{Tr}}(\mathbf{Q}_{k,(2)}\mathbf{V}_n)}\zeta_{k,(2)}\leq{	\frac{\alpha_{k}}{2}\zeta_{k,(2)}^2+\frac{1}{2\alpha_{k}}{\rm{Tr^2}}(\mathbf{Q}_{k,(2)}\mathbf{V}_n)},
\end{equation}
where $\alpha_{k}$ is a fixed point and the equality of (\ref{re3}) holds when $\alpha_{k}=\frac{\zeta_{k,(2)}}{{\rm{Tr}}(\mathbf{Q}_{k,(2)}\mathbf{V}_n)}$.

By invoking the CUB approximation, constraint (\ref{re2}) can be approximated as:
\begin{equation}\label{re4}
	\frac{\alpha_{k}}{2}\zeta_{k,(2)}^2+\frac{1}{2\alpha_{k}}{\rm{Tr^2}}(\mathbf{Q}_{k,(2)}\mathbf{V}_n)+{\sigma_k^2}\zeta_{k,(2)}\leq{p_{k,(2)}}{{\rm{Tr}}(\mathbf{Q}_{k,(2)}\mathbf{V}_n)}.
\end{equation}

Constraint (\ref{re4}) is then convex and the final subproblem of optimizing transmission and reflection coefficient vector is written as:
\begin{subequations}\label{p5}
	\begin{align}
		\mathop {\max}\limits_{\left\{\boldsymbol{\zeta},\mathbf{V}_n \right\}} &\;\;\sum_{k=1}^{K}\sum_{i=1}^{2}\log_2({1+\zeta_{k,(i)}})  \\
		{\rm{s.t.}}\;\;&\rm(\ref{QoS}),(\ref{energy_con2}),(\ref{psd}),(\ref{sic2}),(\ref{re4}).
	\end{align}
\end{subequations}

Problem (\ref{p5}) needs to be solved iteratively using CUB approximation until convergence. More specifically, we first initialize $\alpha_{k}^0$ based on any feasible points $\zeta_{k,(2)}^0$ and $\mathbf{V}_n^0$. Then for step $t$, we solve problem (\ref{p5}) and update $\alpha_{k}^{t}=\frac{\zeta_{k,(2)}^{t-1}}{{\rm{Tr}}(\mathbf{Q}_{k,(2)}\mathbf{V}_n^{t-1})}$ until the increase of the objective function is within the predefined accuracy $\epsilon$. It is noted that for each iteration, the problem (\ref{p5}) is convex and can be 
efficiently solved by standard toolboxes like CVX\cite{cvx}. 

%
%However, randomly generated initial points $\zeta_{k,i}^0,\mathbf{V}_n^0$ may cause the problem (\ref{p5}) infeasible. To tackle this issue, we propose a feasible point searching algorithm. Specifically, before solving problem (\ref{p5}), we first solve the following feasibility error-minimization problem
%\begin{subequations}\label{p10}
%	\begin{align}
%	\mathop {\min}\limits_{\left\{\boldsymbol{\zeta},\mathbf{V}_n,z \right\}} &\;\;z \\
%		{\rm{s.t.}}\;\;
%		&\gamma-z\leq\sum_{k=1}^K\log_2(1+\zeta_{k,i})\\
%		&\begin{aligned}
%					&	\frac{\alpha_{k,i}}{2}\zeta_{k,i}^2+\frac{1}{2\alpha_{k,i}}{\rm{Tr^2}}(\mathbf{Q}_{k,i}\mathbf{V}_n)+\\
%			&{\sigma_k^2}\zeta_{k,i}-z\leq{P_{k,i}}{{\rm{Tr}}(\mathbf{Q}_{k,i}\mathbf{V}_n)}.
%		\end{aligned}\\
%		&\rm(\ref{unitmo2}),(\ref{psd}).
%	\end{align}
%\end{subequations} 
%where $z$ is an introsuced non-negative auxiliary variable to enlarge the feasible set. When $z$ equals 0, all the constraints of problem (\ref{p10}) is the same as those of (\ref{p5}), so as the feasible set. As a result, the proposed feasibility searching method can be solved primarily before designing the beamforming matrix, and the results can be utilized as feasible initial points to start the algorithm. Problem (\ref{p10}) is more robust and tolerable than problem (\ref{p5}) since the initial point can be randomly generated. 

\begin{algorithm}[t]
	\caption{Transmission- and Reflection-Coefficient Vector Optimization Algorithm}
	\label{alg2}
	\begin{algorithmic}[1]
		\STATE \textbf{Initialize} Initialize channel assignment scheme $\lambda_{k,i}$ and decoding order $\pi_k(i)$, tolerance $\epsilon$. 
		\STATE \qquad Set $t=0$. Generate initial feasible points $\alpha_{k}^t$.% by solving problem (\ref{p10})
		\STATE \qquad Compute $R_{\text{sum}}^t=\sum_{k=1}^{K}\sum_{i=1}^{2}R_{k,(i)}^t$. 
		\STATE \textbf{repeat}
		\STATE \qquad Solve problem (\ref{p5}) to obtain $\mathbf{V}_n^{t+1}$ and $\zeta^{t+1}_{k,(i)}$.
		\STATE \qquad Update $\alpha_{k}^{t+1}$ as $\frac{\zeta_{k,(2)}^{t+1}}{{\rm{Tr}}(\mathbf{Q}_{k,(2)}\mathbf{V}_n^{t+1})}$.	
		\STATE \qquad Compute $R_{\text{sum}}^{t+1}$. 
		\STATE \qquad $t=t+1$.		
		\STATE \textbf{until} $\left|R_{\text{sum}}^{t+1}-R_{\text{sum}}^{t}\right|\leq{\epsilon}$.
	\end{algorithmic}
\end{algorithm}
The process of optimizing transmission- and reflection-coefficient vectors is summarized in \textbf{Algorithm 3}.

\subsection{Power Allocation}
With derived decoding order and beamforming vector, the power allocation can be obtained by solving the following problem
\begin{subequations}\label{p6}
	\begin{align}
		\mathop {\max}\limits_{\left\{\mathbf{p} \right\}} &\;\;\sum_{k=1}^{K}\sum_{i=1}^{2}R_{k,(i)}  \\
		{\rm{s.t.}}\;\;&\rm(\ref{QoS}),(\ref{positvepower}),(\ref{totalpower}).
	\end{align}
\end{subequations}

Note that with given beamforming vector $\mathbf{v}_n$, the channel power gains $\left|h_{k,(i)}\right|^2$ are settled. For ease of expression, we denote $m_{k,(i)}=\frac{\sigma_k^2}{|h_{k,(i)}|^2}$ as the normalized noise power of the $i$-th user on sub-channel $k$, $i=1,2$. Based on the determined decoding order, we always have $m_{k,(2)}>m_{k,(1)}$. Next, we will reformulate problem (\ref{p6}) as geometric programming utilizing the following lemma. 
\begin{lemma}\label{lemma2}
	By a change of varibales, problem (\ref{p6}) can be converted to the following GP problem
\begin{subequations}\label{p7}
	\begin{align}
				\mathop {\max}\limits_{\left\{r_{k,(i)} \right\}} &\;\;\sum_{k=1}^{K}\sum_{i=1}^{2}\log_2({r_{k,(i)}})  \\
		\label{gp_power}{\rm{s.t.}}\;\;&r_{k,(i)}^{-1}\leq1,\\
		\label{gp_qos}&\prod_{k=1}^K{r_{k,(i)}^{-1}}\leq2^{-\gamma},\\
		\label{gp_totalpower}&\sum_{k=1}^K\frac{m_{k,(1)}r_{k,(1)}r_{k,(2)}+(m_{k,(2)}-m_{k,(1)})r_{k,(2)}}{{P_{\text{max}}}+m_{k,(2)}}
		\leq1,
	\end{align}
\end{subequations}
where $r_{k,(i)}$ is the new variable which is equal to $2^{R_{k,(i)}}$. 

The corresponding power allocation $p_{k,(i)}$ can be derived by solving the equations (\ref{opt_p}).

\end{lemma}
 \begin{proof}\renewcommand{\qedsymbol}{}
 	\emph{See Appendix B.}	
 \end{proof}
The objective function of problem (\ref{p7}) is convex, and the left-hand side parts of constraints (\ref{gp_power})-(\ref{gp_totalpower}) are all posynomials (or monomials)\cite{gp}. Thus, problem (\ref{p7}) is  consistent with the GP form. Note that GP can be equivalently converted to a convex optimization problem \cite{gp}, whose optimal solution can be obtained using toolboxes such as CVX\cite{cvx}. 

The proposed three-step algorithm for optimizing decoding order, beamforming vector, and power allocation is summarized in \textbf{Algorithm 4}.

%\begin{proposition}
%When the total power budget is equally allocated to each sub-carrier, i.e., constraint (\ref{totalpower}) is given by $\sum_{i=1}^{I}p_{k,i}\leq{P_{\text{max}}}/K$. The power allocation scheme can be given as closed-form solutions such that
%\begin{equation}
%	P_{\pi_k^{-1}(i)}^*=(2^{\gamma}-1)(m_{\pi_k^{-1}(i)}+\sum_{j>i}P_j^*)
%\end{equation}
%\end{proposition}
%\begin{proof}\renewcommand{\qedsymbol}{}
%	\emph{See appendix C.}
%\end{proof}
\subsection{Convergence and Complexity Analysis}
\begin{algorithm}[t]
	\caption{Proposed Three-Step Algorithm for NOMA}
	\label{alg3}
	\begin{algorithmic}[1]
		\STATE \textbf{Initialize} Initialize channel assignment scheme $\lambda_{k,i}$ via \textbf{Algorithm 1}, tolerance $\epsilon$. 
		\STATE \textbf{Step 1} Design of decoding order
		\STATE \qquad Solve problem (\ref{pro_max_channel_gain2}) to obtain the channel gain of each user $|h_{k,i}|^2$.
		\STATE \qquad Obtain $\pi_{k}(i)$ according to derived $|h_{k,i}|^2$.
		\STATE \textbf{Stage 2} Optimization of beamforming vector 
		\STATE \qquad Optimize the transmission and reflection beamforming vector via \textbf{Algorithm 3}.
		\STATE \textbf{Stage 3} Optimal power allocation
		\STATE \qquad Solve problem (\ref{p7}) using GP.
		\STATE \textbf{Output} Converged solution $\boldsymbol{\pi}^*,\mathbf{\Theta}^*,\boldsymbol{p}^*$.
	\end{algorithmic}
\end{algorithm}
%Based on the proposed algorithms, the decoding order, beamforming matrix and power allocation are optimized step-by-step. %The details of solving the first subproblem of (\ref{problem}) is summarized in \textbf{Algorithm 1}.

\emph{1) Convergence Analysis:} 
For LMA, since it is a special case of the swap operation-enabled algorithm, the algorithm is guaranteed to converge to a two-sided exchange stable matching until there is no swap-blocking pair. Then, the convergence of the proposed three-step algorithm mainly depends on \textbf{Algorithm 3} since only the second step requires iteration. Next, we will prove the convergence of Step 2. Define $U(\boldsymbol{\Theta}^{t_0})$ and $U(\boldsymbol{\Theta}^{t_0},\boldsymbol{\zeta}^{t_0})$ as the objective function’s value of
problem (\ref{p3}) and problem (\ref{p4}) at the $t_0$-th iteration, respectively. We have
\begin{equation}
	\begin{aligned}	U(\boldsymbol{\Theta}^{t_0})&\overset{(a)}{=}U(\boldsymbol{\Theta}^{t_0},\boldsymbol{\zeta}^{t_0})\overset{(b)}{=}U(\boldsymbol{\Theta}^{t_0},\boldsymbol{\zeta}^{t_0},\boldsymbol{\alpha}^{t_0})\\
		&\overset{(c)}{\leq}U(\boldsymbol{\Theta}^{t_0},\boldsymbol{\zeta}^{t_0+1},\boldsymbol{\alpha}^{t_0})\overset{(d)}{=}U(\boldsymbol{\Theta}^{t_0},\boldsymbol{\zeta}^{t_0+1},\boldsymbol{\alpha}^{t_0+1})
	\end{aligned}
\end{equation}
where (a) is due to the equivalent transformation of problem, which is proved in \textbf{Proposition 2}. (b) and (d) hold since we use the CUB to replace the non-convex parts, which is updated by the last iteration. (c) comes from that $\boldsymbol{\zeta}$ is the only variable to be optimized with other variables fixed. Thus, after each iteration, the objective function is non-decreasing. Since the system sum rate is upper-bounded, the overall algorithm for Step 2 is guaranteed to converge.
%Thus, the objective value of \textbf{Algorithm 2} is monotonically non-decreasing.

\emph{2) Complexity Analysis:}  For LMA, since the swap operations only occur among users belonging to the same region, the maximum number of swap operations is reduced to $\frac{I/2(I/2-1)}{2}$. The sub-problem of decoding order design is solved using SDP with complexity $\mathcal{O}(M)^{4.5}$. 
The complexity of Algorithm 3 to optimize transmission- and reflection-coefficient matrices is $\mathcal{O}(I_{\text{itr}}(2M+I)^{4.5}\log(1/\epsilon))$, where $I_{\text{itr}}$ is the iterations taken to converge, $\epsilon$ is the accuracy\cite{sdr}. The complexity of solving the GP of power allocation using interior-point methods is given by $\mathcal{O}(I)^{3.5}$\cite{convex}.

\section{Numerical Results}

In this section, numerical results are provided to validate our proposed designs. As illustrated in Fig. \ref{setup}, we consider a three-dimensional (3D) coordinate system, where the AP is located at the origin and the STAR-RIS is located in $x$-axis. The reference center of the STAR-RIS is set at (50, 0, 0) meters. The served users are uniformly and randomly distributed on a circle centered at the STAR-RIS with a radius of $r=$ 5 m. The distance-dependent path loss is modeled as $L(d)=\rho_0(\frac{d}{d_0})^{-\alpha}$, where $d$ is the individual link distance, $\rho_0$ denotes the path loss at the reference distance $d_0$ = 1 m, and $\alpha$ denotes the path loss exponent. Then, the Rician fading subchannels from the AP to the STAR-RIS and from the STAR-RIS to users are given by
\begin{figure}[t!]
	\centering
	\begin{minipage}[t]{0.45\linewidth}
		\includegraphics[width=3in]{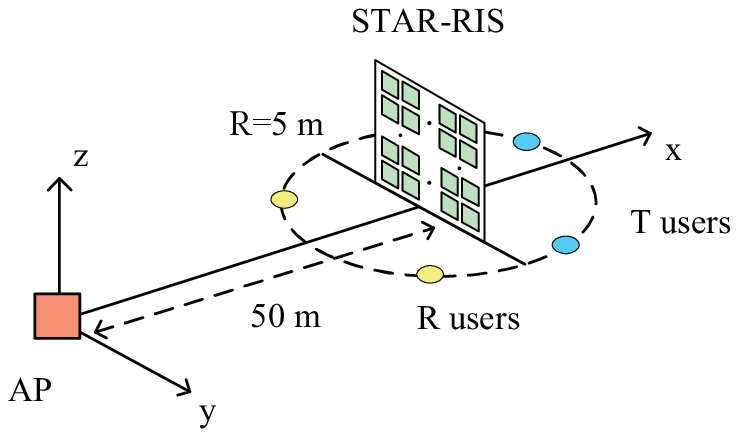}
		\caption{Simulation setup.}
		\label{setup}
	\end{minipage}
	\quad
	\begin{minipage}[t]{0.45\linewidth}
		\includegraphics[width=2.6in]{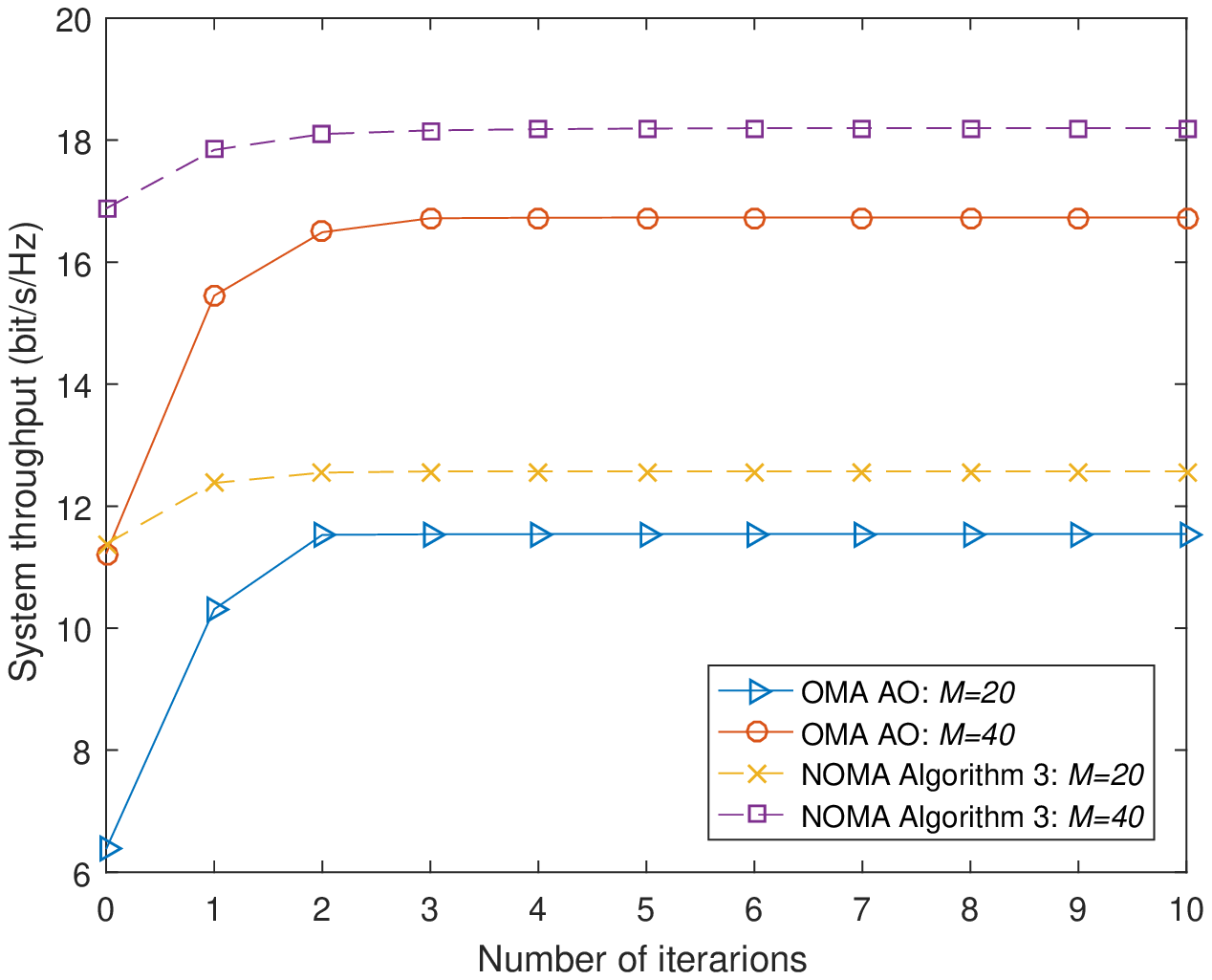}
		\caption{Convergence behavior of proposed AO for OMA and Algorithm 3 for NOMA. $I=6.$}
		\label{con}
	\end{minipage}
\end{figure}

\begin{subequations}
	\begin{align}
			&{\mathbf{g}_k} = \sqrt {L(d_{AS})} \left( {\sqrt {\frac{{{\kappa}_{AS}}}{{{\kappa}_{AS} + 1}}} {\mathbf{g}}_k^{{\rm{LoS}}} + \sqrt {\frac{1}{{{\kappa}_{AS} + 1}}} {\mathbf{g}}_k^{{\rm{NLoS}}}} \right),\\
			&{{\mathbf{f}}_{k,i}} = \sqrt {L(d_{SU,i})} \left( {\sqrt {\frac{{{\kappa}_{SU}}}{{{\kappa}_{SU} + 1}}} {\mathbf{f}}_{k,i}^{{\rm{LoS}}} + \sqrt {\frac{1}{{{\kappa}_{SU} + 1}}} {\mathbf{f}}_{k,i}^{{\rm{NLoS}}}} \right),
	\end{align}
\end{subequations}
where $d_{AS}$ and ${d_{SU,i}}$ denote the distance between the AP and the STAR-RIS and between the STAR-RIS and user $i$, respecively, ${{\kappa_{AS}}}$ and $\kappa_{SU}$ denote the Rician factor, ${{\mathbf{g}}_k^{{\rm{LoS}}}}$ and ${{\mathbf{f}}_{k,i}^{{\rm{LoS}}}}$ are the deterministic line-of-sight (LoS) components, ${{\mathbf{g}}_k^{{\rm{NLoS}}}}$ and ${{\mathbf{f}}_{k,i}^{{\rm{NLoS}}}}$ are the non-line-of-sight (NLoS) components modeled as Rayleigh fading. In this paper, the path loss exponents for the AP-STAR-RIS link and STAR-RIS-user link are set to be $\alpha_{AS}=$ 2.2, $\alpha_{SU}=$ 2.8, respectively, the Rician factors are $\kappa_{AS}=\kappa_{SU}=3$ dB, and the noise power is set to be $\sigma_k=-80$ dBm\cite{simulation}. The common minimum rate requirement is given by $\gamma=0.1$ bit/s/Hz. The tolerance of convergence $\epsilon$ is set to be $10^{-4}$. The maximum transmit power at the AP is set as $P_{\text{max}}=1.5$ W, unless stated otherwise.

%\begin{figure}[t]
%	\centering
%	\includegraphics[width=0.45\textwidth]{eps/system_setup2.eps}\\
%	\caption{Simulation setup}\label{setup}
%\end{figure}
%\begin{figure}[t]
%	\centering
%	\subfloat[Covergence behavior of the first stage of algorithm 2. $K=4.$]{		
%		\includegraphics[width=0.4\textwidth]{eps/convergence1.eps}	
%	}
%	\hspace{0.1em}
%	\subfloat[Covergence behavior of the overall algorithm. $M=20.$]{		
%		\includegraphics[width=0.4\textwidth]{eps/convergence2.eps}
%	}	
%	\caption{Convergence behavior of the proposed algorithms.}
%	\label{con}
%\end{figure}
\subsection {Convergence Behavior of Proposed Algorithms}

%\begin{figure}[htbp]
%	\centering
%	\subfigure[pic1.]{
%		\includegraphics[width=5.5cm]{111.eps}
%		%\caption{fig1}
%	}
%	\subfigure[pic2.]{
%		\includegraphics[width=5.5cm]{111.eps}
%	}
%	\quad
%	\subfigure[pic3.]{
%		\includegraphics[width=5.5cm]{111.eps}
%	}
%	\subfigure[pic4.]{
%		\includegraphics[width=5.5cm]{111.eps}
%	}
%	\caption{ pics}
%\end{figure}

\begin{figure}[t!]
	\centering
	\begin{minipage}[t]{0.45\linewidth}
		\includegraphics[width=2.8in]{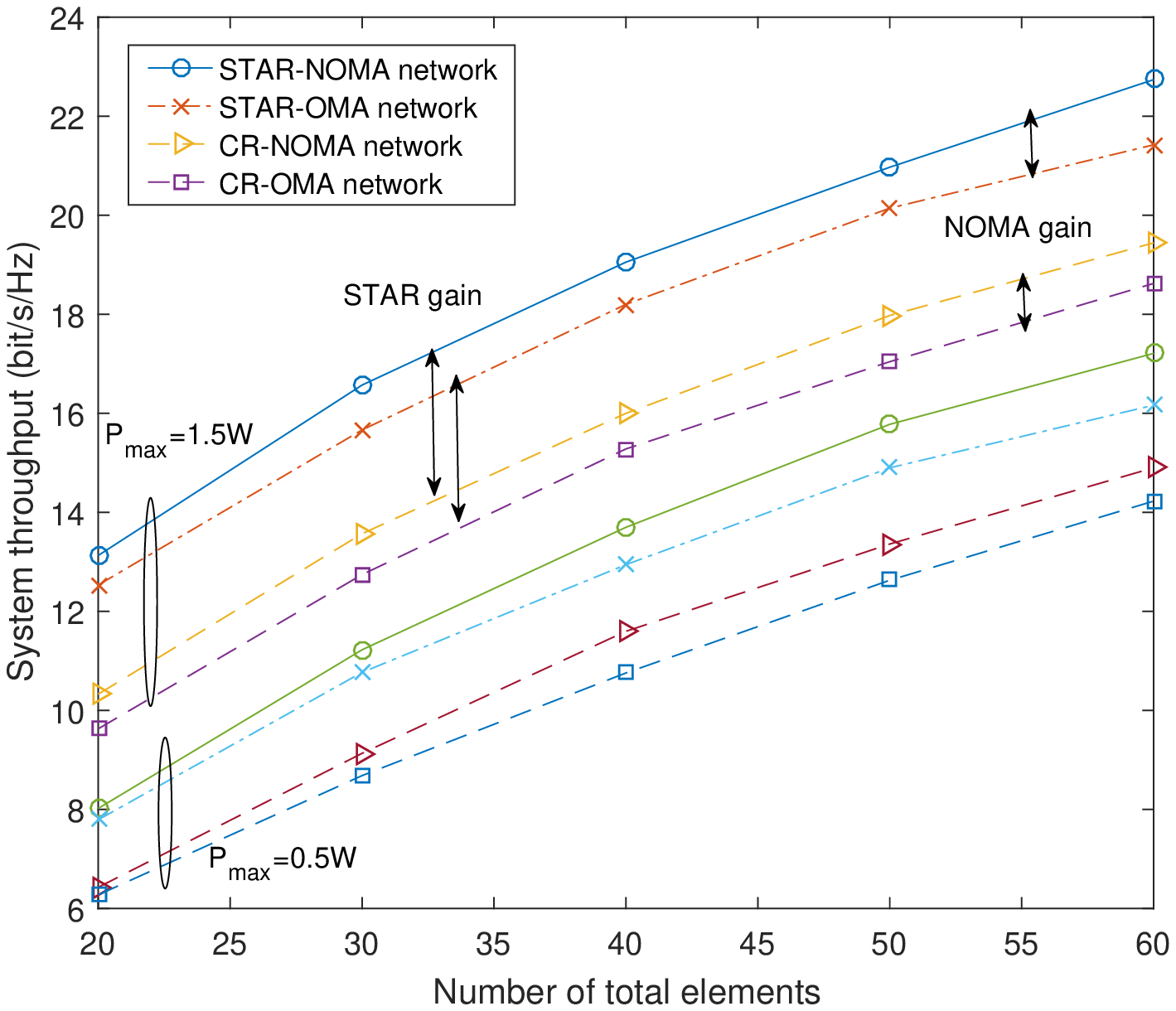}
		\caption{The system sum-rate versus the number of total elements $M$ with different energy budget. $I=6,K=3.$}
		\label{m1}
	\end{minipage}
	\quad
	\begin{minipage}[t]{0.45\linewidth}
		\includegraphics[width=3in]{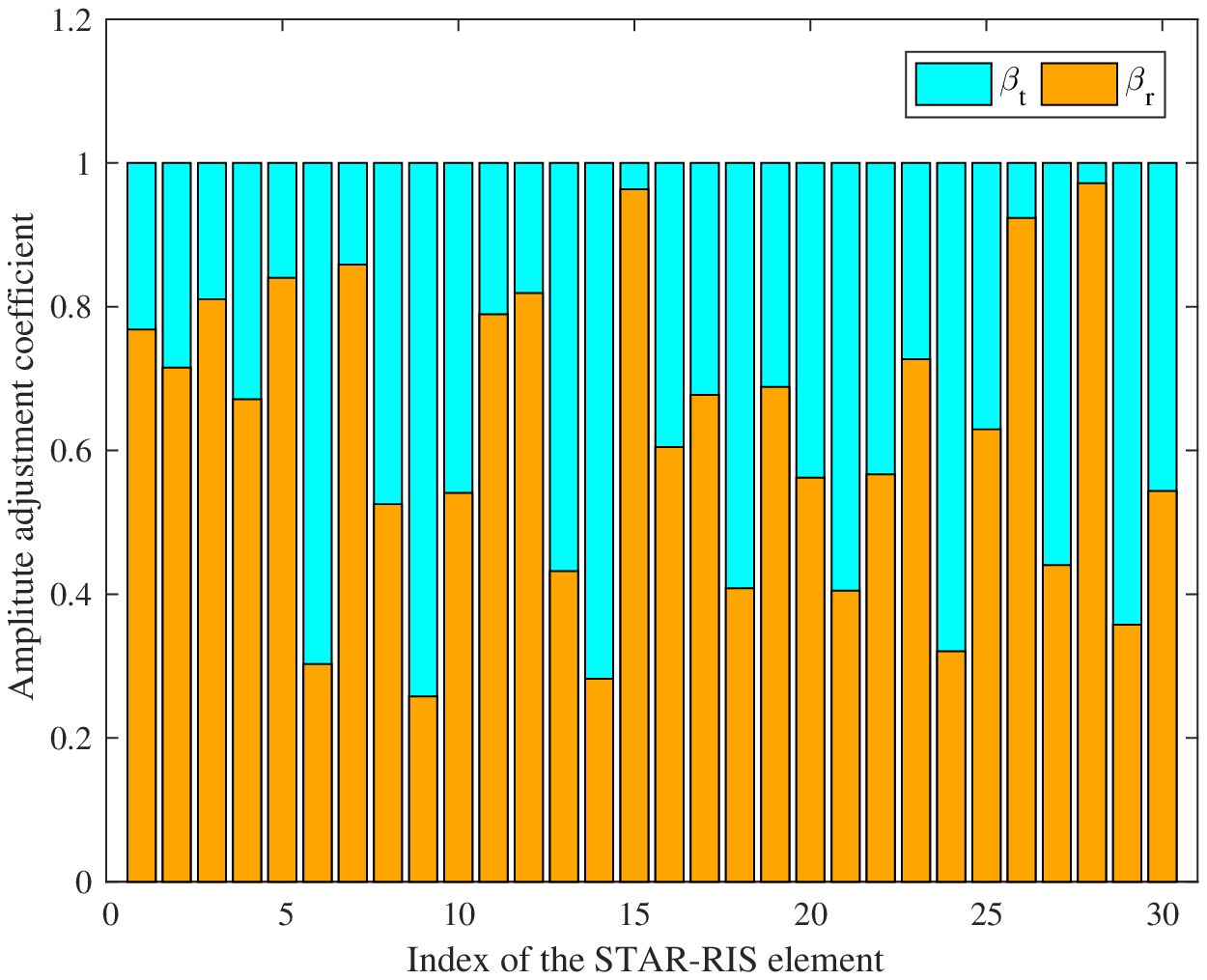}
		\caption{Amplitude adjustment of transmitted and reflected signals by each element for NOMA. $I=6,K=3$.}
		\label{amplitude}
	\end{minipage}
\end{figure}
In Fig. \ref{con}, the convergence behavior of the proposed algorithms is illustrated using numerical results. Specifically, for OMA, we demonstrate the convergence behavior of proposed AO algorithm. It is observed from the figure that the proposed AO algorithm for OMA converges fast, mostly with 3 iterations. For NOMA, we plot the convergence of Algorithm 3 since we use CUB to approximate the optimal solution through iterations. From Fig \ref{con}, it is found that the number of iterations required for the convergence of Algorithm 3 increases with $M$ since the complexity increases exponentially with larger-scale beamforming coefficient matrices to be optimized.
%The convergence of the overall Algorithm 1 with different $K$ is shown in Fig \ref{con}(b), from which we can observe that the steps taken to converge increase with the number of users (subchannels) since more swap operations are executed. From the figure we can see that when $I=10$ and $K=5$, Algorithm 1 requires 5 steps iterations to converge, and there requires over 10 swap operations behind on average.

%\begin{figure}   
%	\begin{minipage}[t]{0.5\linewidth} % 如果一行放2个图，用0.5，如果3个图，用0.33  
%		\centering   
%		\includegraphics[width=1in]{eps/convergence1.eps}   
%		\caption{Small Box}   
%		\label{fig:side:a}   
%	\end{minipage}%   
%	\begin{minipage}[t]{0.5\linewidth}   
%		\centering   
%		\includegraphics[width=1.5in]{eps/convergence2.eps}   
%		\caption{Big Box}   
%		\label{fig:side:b}   
%	\end{minipage}   
%\end{figure}

\subsection{Impact of the STAR-RIS}
To demonstrate the benefits brought by STAR-RIS, we consider the following baseline schemes:
\begin{itemize}
%	\item \textbf{STAR-RIS assisted OMA System:} In this case, we consider OMA transmission protocol such that frequency and time division multiple access are considered to serve multiple users simultaneously. 	
	\item \textbf{Conventional RIS-assisted OMA/NOMA system (also referred to as CR OMA/NOMA system):} In this case, one traditional reflecting-only and one transmitting-only RIS are deployed adjacently at the same place as the STAR-RIS to achieve full-space coverage. For a fair comparison, each CR is equipped with $M/2$ elements. Note that the resource allocation problem in CR-aided multi-carrier system has been studied in \cite{nomaris3}.
\end{itemize}

%\begin{figure}[t]
%	\centering
%	\includegraphics[width=0.5\textwidth]{eps/rate_M.eps}\\
%	\caption{The system sum rate versus the number of total elements $M$ with different energy budget. $I=6,K=3.$}\label{m1}
%\end{figure}
%
%\begin{figure}[t]
%	\centering
%	\includegraphics[width=0.5\textwidth]{eps/amplitude.eps}\\
%	\caption{Amplitude adjustment of transmitted and reflected signals by each element. $I=6,K=3,P_\text{max}=0.6$ W.}\label{amplitude}
%\end{figure}
Fig. \ref{m1} shows the system achievable sum-rate versus the total number of elements $M$ with different transmit power budgets. Firstly, it is observed that the sum rates of all considered systems increase dramatically with more total elements since they benefit from higher beamforming gain. Secondly, the sum-rate of the STAR-NOMA network outperforms that of the STAR-OMA networks due to the multiplexing gain. Note that since the users are located at the same distance from the STAR-RIS, the original cascaded channel condition of each user is similar, which is undesirable for NOMA transmission. However, with the STAR-RIS, the NOMA gain is significantly enhanced since the channel conditions of different users are reconfigured to be more distinct. %It is worth mentioning that the OMA-II scheme yields great performance gain compared with OMA-I, which reveals the importance of optimizing the time allocation factor. 
Finally, the STAR-aided networks also outperform conventional RIS-aided networks, and the performance gap increases as there are more elements. This is because STAR-RIS can make full use of the DoFs to manipulate the signal propagation. Reversely, for CR, the performance gain is limited since there are fewer coefficients to be optimized. %Actually, the optimization problem for CR serves as a special case of STAR-RISs, where each metasurface only works at single-mode (transmission or reflection). %Moreover, by properly design the transmission and reflection beamforming coefficient, the STAR-RIS can enlarge the channel disparities to further increase the gain of NOMA compared with CR NOMA system.

Fig. \ref{amplitude} plots the amplitude adjustment of the transmitted and reflected signals for each STAR-RIS element in a NOMA network. In this case, we consider a scenario that 6 users are located symmetrically on the two regions. The NOMA gain may not be fully reaped if the network is assisted by conventional RISs since the channel gains of served users approximate each other.
However, STAR-RIS can well overcome this drawback by adjusting the amplitudes of incident signals to allocate more power to the side with more `strong users'. As shown in Fig. \ref{amplitude}, the STAR-RIS allocates more energy to the reflection region after user pairing, since two users with higher decoding orders are located in this region. Thus, more users can benefit from the NOMA transmission.

\subsection{Impart of Channel Assignment}
%\begin{figure}[t]
%	\centering
%	\includegraphics[width=0.5\textwidth]{eps/user46.eps}\\
%	\caption{Empirical CDF of sum rate performance under different channel assignment schemes with $I=4$ and $I=6$. $P_\text{max}=0.15$ W.}\label{cdf}
%\end{figure}

%\begin{minipage}[t]{0.45\linewidth}
%	\includegraphics[width=2.8in]{eps/do.eps}
%	\caption{Average system sum-rate with different number of users under different docoding order designs.}
%	\label{do}
%\end{minipage}
To evaluate the effectiveness of our proposed channel assignment schemes for OMA and NOMA, we consider the exhaustive search-based algorithm as a benchmark, which also serves as the performance upper bound. We also consider the complementary user-pairing scheme of LMA as a baseline, referred to as the same-region matching algorithm (SMA), where we preferentially pair users located in the same region on a sub-channel. We plot the cumulative distribution function (CDF) of the system sum rate under different channel assignment schemes with 4 and 6 total users, respectively.

\begin{figure}[t!]
	\centering	
	\begin{minipage}[t]{0.45\linewidth}
		\includegraphics[width=2.8in]{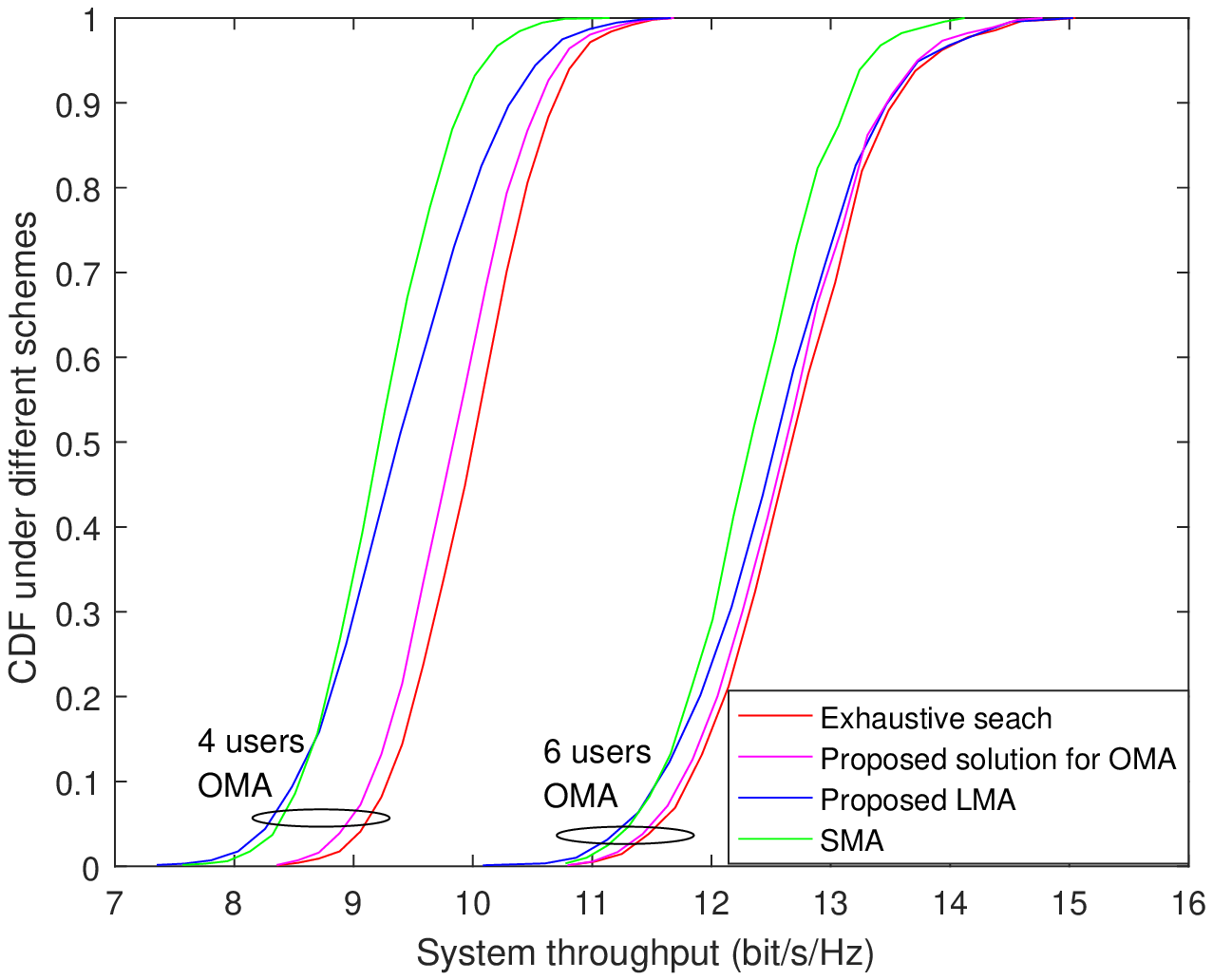}
		\caption{Empirical CDF of sum rate performance under different channel assignment schemes for OMA with $I=4$ and $I=6$.}
		\label{cdf_oma2}
	\end{minipage}
	\quad
	\begin{minipage}[t]{0.45\linewidth}
		\includegraphics[width=2.8in]{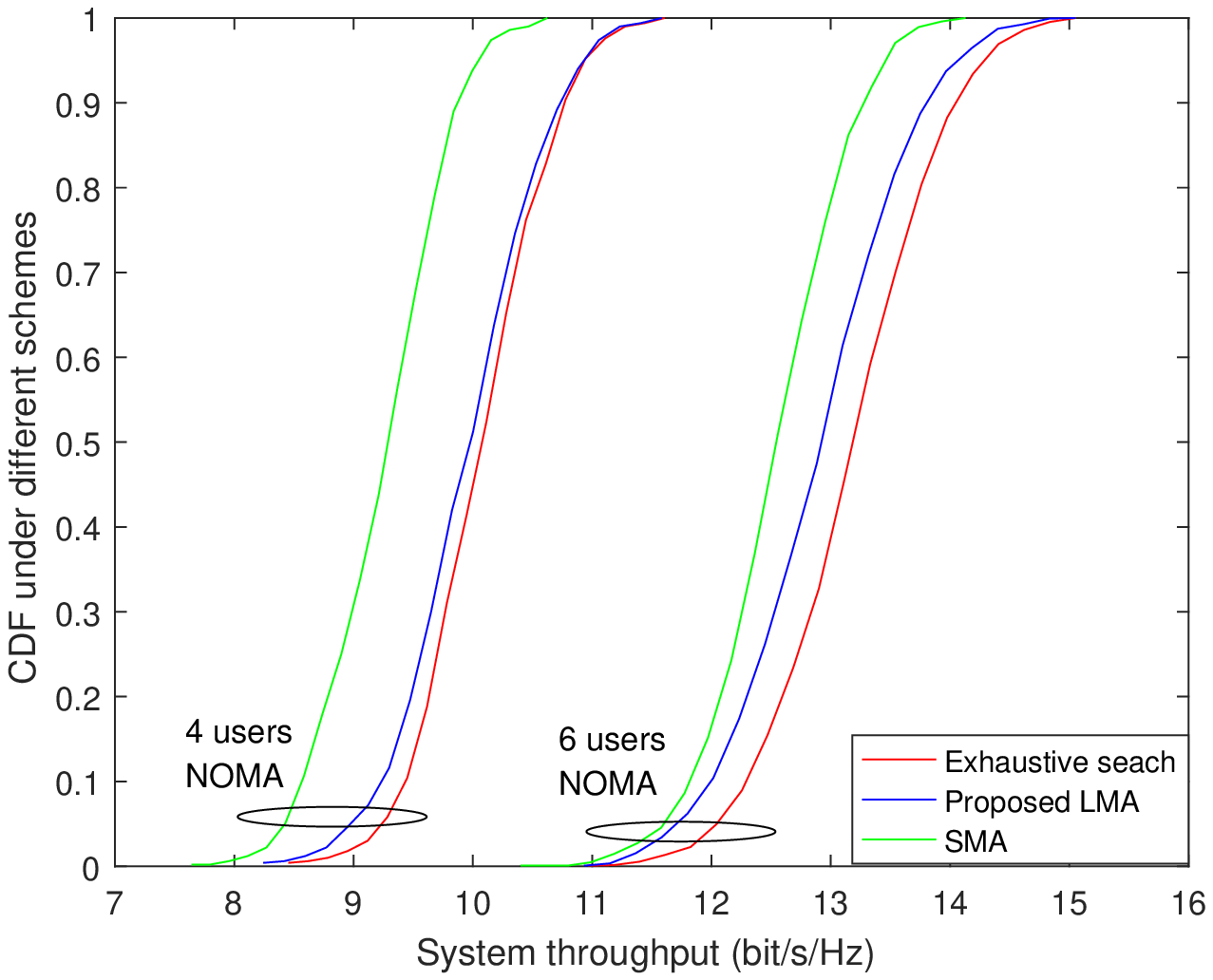}
		\caption{Empirical CDF of sum rate performance under different channel assignment schemes for NOMA with $I=4$ and $I=6$.}
		\label{cdf}
	\end{minipage}
\end{figure}

Fig. \ref{cdf_oma2} evaluates the CDF performance of our proposed channel assignment algorithm for OMA. It is observed from Fig. \ref{cdf_oma2} that the proposed solution of channel assignment for OMA achieves near-optimal behavior. To obtain more design insights, we also adopt LMA and SMA schemes for comparison. For the 6-user setting, LMA could also achieve comparable performance as the exhaustive search-based algorithm. In this case, LMA outperforms SMA for most of the time. However, for the setting with 4 users, LMA and SMA both have the probability of reaching the optimal performance as the exhaustive search-based algorithm. This inspires us that the transmission-reflection user-pairing scheme is not always the optimal choice for OMA transmission. Instead, we ought to consider all possibilities of channel assignment to achieve a good performance, i.e., also consider grouping users that are located on the same side of the STAR-RIS.

Fig. \ref{cdf} illustrates the effectiveness of our proposed channel assignment design for NOMA. For the 4-user setting with 2 subchannels, since there are 4 different channel assignment choices for LMA but only 2 choices for SMA, LMA is analytically better than SMA, which is verified by the simulation results. For the 6-user setting, the total number of channel assignment choices of LMA and SMA is the same, which is 48. However, LMA still yields better performance than SMA.
It is observed from Fig. \ref{cdf} that LMA has a high probability of having the same sum rate as the exhaustive search-based algorithm, while SMA can hardly reach the optimal solution. This is because the transmission-reflection user pairing scheme can make full use of the DoF that the STAR-RIS brings to wireless communication networks. Specifically, the manipulation of amplitudes is dominant and effective for channel reconfiguration, so that NOMA benefits a lot from channel disparities that the STAR-RIS provides. This insight inspires us to employ the transmission-reflection user pairing scheme for NOMA transmission, which achieves near-optimal behavior and can be performed with low complexity.  

\subsection{Impact of NOMA Decoding Orders}

%\begin{figure}[t]
%	\centering
%	\includegraphics[width=0.45\textwidth]{eps/do.eps}\\
%	\caption{Average system sum-rate with different number of users under different dosing order designs. $P_\text{max}=0.15$ W.}\label{do}
%\end{figure}
Fig. \ref{do} depicts the impact of the proposed decoding order designs on the system rate performance with different numbers of users and total elements $M$. To evaluate our purposed schemes, three benchmark algorithms for decoding order design are considered, which are exhaustive search algorithm, random initialization algorithm, and cascaded channel-based algorithm respectively. For exhaustive search, we search over all different decoding orders and the results serve as the performance upper bound. For the random initialization algorithm, we generate arbitrary decoding orders for each sub-channel. We also consider the calculated decoding orders based on the transmission- and reflection-coefficient vectors given in (\ref{ini_phase}), referred to as the cascaded channel-based algorithm. It is observed from Fig. \ref{do} that the system sum-rate deteriorates under the random decoding orders for all settings, which indicates the importance of finding the proper decoding orders. With no more than 6 users, the cascaded channel-based algorithm and our proposed maximum combined channel gain-based algorithm can achieve near-optimal performance, while the performance of cascaded channel-based algorithm starts to decrease when $I=8$ and $M=30$. It is deducted that the cascaded channel-based algorithm is only suitable for small-scale networks since the STAR-RIS may adjust the incident signals by a large margin and the original channel superiority is changed.

\begin{figure}[t!]
	\centering
	\includegraphics[width=0.45\textwidth]{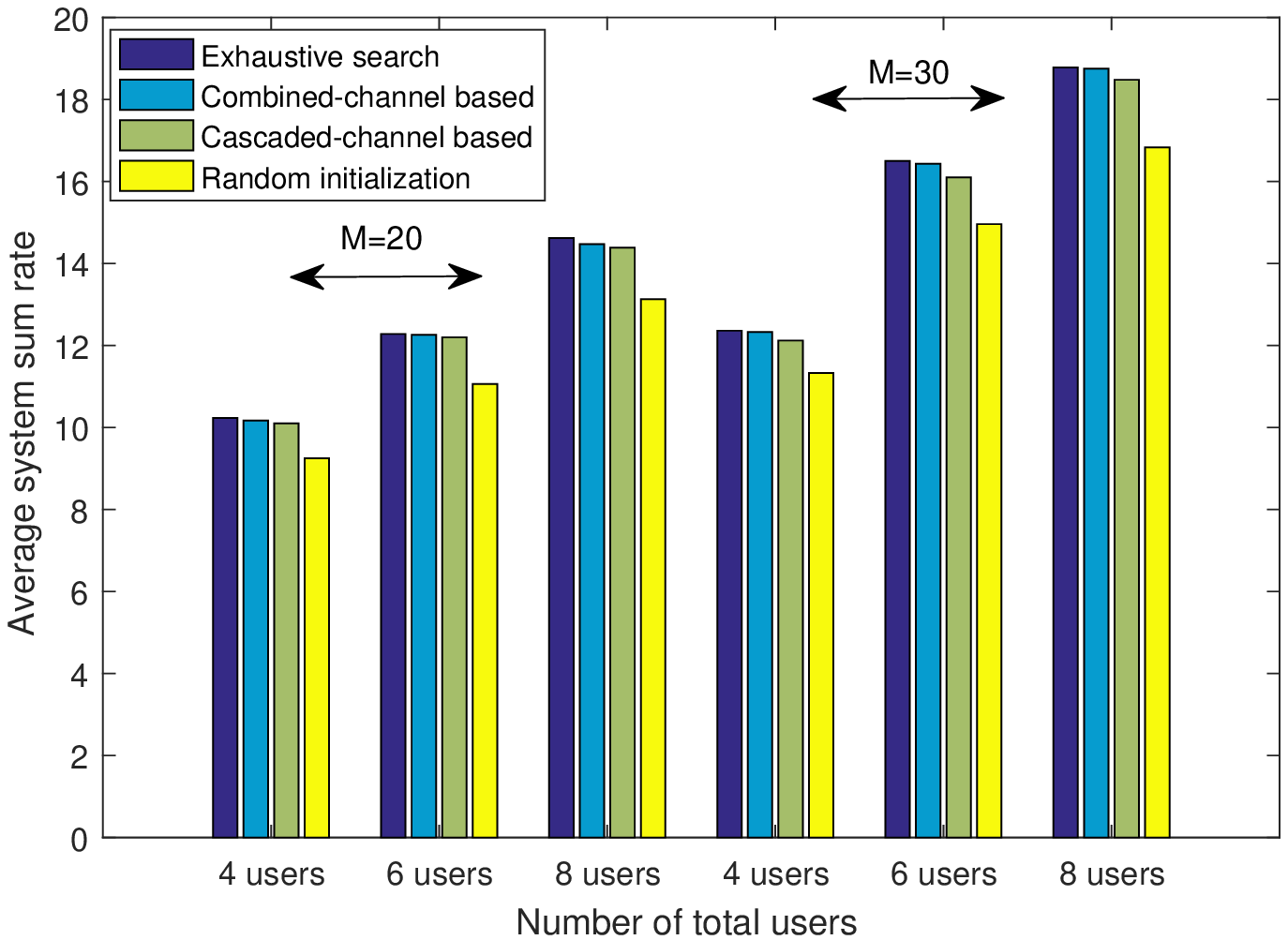}\\
	\caption{Average system sum-rate with different number of users and elements under different decoding order designs.}\label{do}
\end{figure}

%\begin{figure}[t]
%	\centering
%	\includegraphics[width=0.42\textwidth]{eps/convergence1.eps}\\
%	\caption{Convergence behavior of Algorithm 2 Stage 1.}\label{conv1}
%\end{figure}

\section{Conclusions}
In this paper, we have investigated the resource allocation scheme in a STAR-RIS-aided OMA and NOMA network, where the transmission-reflection beamforming at the STAR-RIS and the resource allocation policy at the AP have to be jointly optimized to maximize the sum achievable rate. To tackle the MINLP and NP-hard problem for OMA, we first proposed an efficient scheme to determine the channel assignment using matching theory. Then, we proposed an alternating optimization-based algorithm to optimize the power and time allocation, and beamforming vectors iteratively. For NOMA, a matching algorithm with reduced complexity was proposed to determine the channel assignment, where a transmitted user and a reflected user were grouped on a sub-channel. After channel assignment, we optimized the decoding order, beamforming vector, and power allocation step by step using convex upper bound approximation and geometry programming. Simulation results validated the effectiveness of the proposed designs. In particular, compared with conventional RIS-aided networks and STAR-RIS-OMA networks, STAR-RIS NOMA yielded remarkable performance gains in terms of system sum-rate. Numerical results also inspired us to employ the simple transmission-reflection user-pairing scheme in STAR-RIS-aided multi-carrier NOMA networks. While for OMA, a more comprehensive design that also includes same-side user-pairing was required.
\vspace{-0.2cm}
\section*{Appendix~A: Proof of Theorem~1} \label{Appendix A}
If we prove that a special case of problem (\ref{pro_oma}) and (\ref{problem}) is NP-hard, then the original problems are also NP-hard problems. So we consider the simplest case, where we only optimize the channel assignment and a T user is only paired with an R user on a subchannel.

As defined, $\mathcal{I}_t$ and $\mathcal{I}_r$ are two disjoint sets, $\mathcal{I}_t\cup\mathcal{I}_r=\mathcal{I}$, $\mathcal{I}_t\cap\mathcal{I}_r=\emptyset$ and $|\mathcal{I}_t|=|\mathcal{I}_r|=I/2$. Define subset $\mathcal{D}=\{\mathcal{D}_1,...,\mathcal{D}_K\}$ as a subset of $\mathcal{K}\times\mathcal{I}_t\times\mathcal{I}_r$, where the tuples are $\mathcal{D}_k=(k,i,j)$. The mentioned channel assignment can be stated as: There exists a subset $\mathcal{\widetilde{D}}$ that satisfies: Condition 1) $|\mathcal{\widetilde{D}}|=\min\{K,I/2\}$. Condition 2) For any tuples $(k,i,j)\in\mathcal{D}$ and $(\widetilde{k},\widetilde{i},\widetilde{j})\in\mathcal{\widetilde{D}}$, we have $k\neq\widetilde{k},i\neq\widetilde{i}$ and $j\neq\widetilde{j}$. The above definition is consistent with that of three-dimensional matching, which is NP-hard\cite{di,cui}. 
%We then give the computational complexity in this case. The total number of possible user-pairs is $(I/2-1)!$, while the number of channel assignment schemes for these user pairs is $K!$. Thus, even the special case of the original problems is not expected to be solved in polynomial time.

\vspace{-0.2cm}
\section*{Appendix~B: Proof of Lemma~\ref{lemma2}} \label{Appendix B}
The achievable rate region of user $i$ on subchannel $k$ can be characterized as
\begin{equation}
	\begin{aligned}
		C_{\text{rate}}=\Big\{R_{k,(i)}:R_{k,(1)}\leq\log_2\left(1+\frac{p_{k,(1)}}{m_{k,(1)}}\right), 
		R_{k,(2)}\leq\log_2\Big(1+\frac{p_{k,(2)}}{p_{k,(1)}+m_{k,(2)}}\Big),\forall{k,j}\Big\}.
	\end{aligned}
\end{equation}

When each $R_{k,(i)}$ reaches the boundary of the rate region and by introducing new variables $r_{k,(i)}=2^{R_{k,(i)}}$, the transmit power of each user will satisfy the following equations:
\vspace{-0.2cm}
\begin{equation}\label{opt_p}
	\begin{aligned}
p_{k,(1)}=(r_{k,(1)}-1)m_{k,(1)},\ p_{k,(2)}=(r_{k,(2)}-1)(m_{k,(1)}+p_{m,(1)}), \forall{k}\in\mathcal{K}.
	\end{aligned}
\end{equation} 

	Thus, the constraint (\ref{totalpower}) is equivalent to
	\vspace{-0.2cm}
\begin{equation} \label{gp_appendix_power}
	\begin{aligned}
	\sum_{k=1}^K[m_{k,(1)}r_{k,(1)}r_{k,(2)}+(m_{k,(2)}-m_{k,(1)})r_{k,(2)}]\leq{P_{\text{max}}}+\sum_{k=1}^Km_{k,(2)}.
	\end{aligned}
\end{equation} 

Under optimal decoding order, since the coefficient parts $(m_{k,(2)}-m_{k,(1)})$ are always positive, (\ref{gp_appendix_power}) is a convex constraint.

%\section*{Appendix~C: Proof of Proposition~2} \label{Appendix C}
%Since the power constraint is separate on each sub-carrier, the multi-carrier power allocation problem is equivalent to a single-carrier resource allocation problem. Following the ideas in \cite{power1,power2}, the power allocation scheme with QoS constraints follows the rules that the user with the highest decoding order is allocated as much power as possible while the other users only meet the minimum QoS requirements. Thus, when the channel gain $|h_{k,i}|$ is determined after beamforming design, the power allocated to the user with the lowest decoding order is given by
%\begin{equation}
%	P_{\pi_k^{-1}(N_{\text{max}})}^*=(2^{\gamma}-1)m_{\pi_k^{-1}(N_{\text{max}})}
%\end{equation}
%The user with the second-lowest decoding order will also be allocated minimum power to meet its QoS requirements while regarding the former user's signal as interference.
%Following this rule, the minimum power allocated to each user is given by
%\begin{equation}
%	P_{\pi_k^{-1}(i)}^*=(2^{\gamma}-1)(m_{\pi_k^{-1}(i)}+\sum_{j>i}P_j^*)
%\end{equation}

%\begin{figure}[t!]
%  \centering
%  \includegraphics[width=5in]{eps/converge_tradeoff.eps}\\
%  \caption{Illustration of an example and the expected results.}\label{tradeoff}
%\end{figure}

%\begin{equation}
%	A^{(i)}=\{i+kM|[k\in\{0\}\cup\mathcal{N}]\cap[i+kM\leq{N}]\}
%\end{equation}
\vspace{-0.2cm}
\bibliographystyle{IEEEtran}
\bibliography{reference}

 \end{document}